\documentclass[8pt,letter]{extarticle}
\newcommand{\email}[1]{\href{mailto:#1}{\tt #1}}

\usepackage{verbatim}
\usepackage{amsmath,amssymb,amsthm}
\usepackage{amssymb}
\usepackage{color}
\usepackage{graphicx}
\usepackage{xspace}
\usepackage{dsfont}
\usepackage{hyperref,cite}
\usepackage[numbers]{natbib}
\usepackage{listings}             \newcommand{\id}{\ensuremath{\mathds{1}}}
\newcommand{\reals}{\ensuremath{\mathbb{R}}}
\newcommand{\naturals}{\ensuremath{\mathbb{N}}}
\newcommand{\prob}{\ensuremath{\mathbf{P}}}
\newcommand{\expect}{\ensuremath{\mathbb{E}}}

\newtheorem{lemma}{Lemma}

\newcommand{\catalog}{\mathcal{C}}

\newcommand{\capacity}{\ensuremath{c}}
\newcommand{\source}{\ensuremath{\mathcal{S}}}
\newcommand{\pathset}{\ensuremath{\mathcal{P}}}
\newcommand{\ppath}{\ensuremath{p}}
\newcommand{\requests}{\ensuremath{\mathcal{R}}}
\newcommand{\feasibledomain}{\mathcal{D}}
\newcommand{\argmax}{\mathop{\arg\max\,}}
\newcommand{\argmin}{\mathop{\arg\min\,}}

\newcommand{\trace}{\mathop{\mathtt{trace}}}
\usepackage{algorithm}
\usepackage{algorithmic}

\newcommand{\techrep}[2]{#1\xspace}

\newcounter{packednmbr}
\newenvironment{packedenumerate}{\begin{list}{\thepackednmbr.}{\usecounter{packednmbr}\setlength{\itemsep}{0pt}\addtolength{\labelwidth}{-0pt}\setlength{\leftmargin}{\labelwidth}\setlength{\listparindent}{\parindent}\setlength{\parsep}{0pt}\setlength{\topsep}{3pt}}}{\end{list}}
\newenvironment{packeditemize}{\begin{list}{$\bullet$}{\setlength{\itemsep}{0pt}\addtolength{\labelwidth}{1pt}\setlength{\leftmargin}{\labelwidth}\setlength{\listparindent}{\parindent}\setlength{\parsep}{1pt}\setlength{\topsep}{3pt}}}{\end{list}}

\newcommand{\conv}{\mathtt{conv}}
\newcommand{\supp}{\mathtt{supp}}
\newcommand{\CGS}{\textsc{MaxCG-S}\xspace}
\newcommand{\CGHH}{\textsc{MaxCG-HH}\xspace}
\newcommand{\rCG}{\textsc{R-MaxCG}\xspace}
\newcommand{\iCG}{\textsc{I-MaxCG}\xspace}

\newcommand{\SR}{\ensuremath{\mathtt{SR}}}
\newcommand{\HH}{\ensuremath{\mathtt{HH}}}
\newtheorem{theorem}{Theorem}
\newtheorem{corollary}{Corollary}
\newcommand{\change}[1]{#1}

\title{Jointly Optimal Routing and Caching\\for Arbitrary Network Topologies}
\author{Stratis Ioannidis and Edmund Yeh}
\date{Northeastern University, ECE Dept.\\360 Huntington Avenue, 409DA\\Boston, MA, USA 02115\\\email{ioannidis@ece.neu.edu},\email{eyeh@ece.neu.edu}}

\begin{document}

\maketitle
\begin{abstract}
We study a problem of fundamental importance to ICNs, namely, minimizing routing costs by jointly optimizing caching and routing decisions over an arbitrary network topology. We consider both source routing and hop-by-hop routing settings. The respective offline problems are NP-hard.  Nevertheless, we show that there exist polynomial time approximation algorithms producing solutions  within a constant approximation from the optimal.  We also  produce distributed, adaptive algorithms with the same approximation guarantees.  \change{We simulate our adaptive algorithms over a broad array of different topologies.  Our algorithms reduce routing costs by several orders of magnitude compared to prior art, including algorithms optimizing caching under fixed routing.}
 \end{abstract}

\section{Introduction}
Optimally placing resources in a network and  routing requests toward them is a problem as old as the Internet itself. It is of paramount importance in  information centric networks (ICNs)~\citep{jacobson2009networking,yeh2014vip}, but also naturally arises in a variety of networking applications such as web-cache design~\citep{laoutaris2004meta,che2002hierarchical,zhou2004second},  wireless/femtocell networks \citep{shanmugam2013femtocaching,naveen2015interaction,poularakis2013approximation}, and 
 peer-to-peer networks~\citep{lv2002search,cohen2002replication}, to name a few.
 Motivated by this problem, we study  a \emph{caching network}, i.e., a network  of nodes augmented with additional storage capabilities. In such a network, some nodes act as designated content servers, permanently storing content and serving as  ``caches of last resort''. Other nodes generate requests for content that are forwarded towards these designated servers. If, however, an intermediate node in the path towards a server  stores the requested content, the request is satisfied early: i.e., the request ceases to be forwarded, and a \change{content copy} is sent over the reverse path towards the request's source. 

This abstract setting naturally captures ICNs. Designated servers  correspond to traditional web servers permanently storing content, while nodes generating requests correspond to customer-facing gateways. Intermediate, cache-enabled nodes correspond to storage-augmented routers in the Internet's backbone: such  routers forward requests but, departing from traditional network-layer protocols, immediately serve requests for content they store. An extensive body of research, both theoretical~\citep{che2002hierarchical,fricker2012versatile,martina2014unified,berger2014exact,fofack2012analysis,rosensweig2010approximate,rosensweig2013steady,ioannidis2016adaptive}   and experimental~\citep{lv2002search,laoutaris2004meta,rossi2011caching,zhou2004second,che2002hierarchical,jacobson2009networking}, has focused on modeling and analyzing networks of caches in which \emph{routing is fixed}, and requests follow predetermined paths. For example, shortest paths to the nearest designated server are often used. Given routes to be followed, and the demand for items, the above works aim to model and analyze (theoretically or empirically) the behavior of different caching algorithms deployed over intermediate nodes.

It is not a priori clear whether fixed routing and, more specifically, routing towards the nearest server is the appropriate design choice for such networks. This  is of special interest in the context of ICNs, where delegating routing decisions to another protocol  amounts to  an ``incremental'' deployment. For example, in such a deployment, requests can be forwarded towards the closest designated web servers over paths determined according to, e.g., existing routing protocols such as OSPF or BGP \citep{kurose2007computer}. Subsequent caching decisions by intermediate routers affect only where--within a given path--requests are satisfied. An alternative is to \emph{jointly} optimize both routing \emph{and} caching decisions simultaneously. 
Doing so however poses a significant challenge, precisely because  this joint optimization is inherently combinatorial.
Indeed,  jointly optimizing routing and caching decisions with the objective of, e.g., minimizing routing costs, is an NP-hard problem, and constructing a distributed approximation algorithm is far from trivial~\citep{ioannidis2016adaptive,shanmugam2013femtocaching,fleischer2006tight,borst2010distributed}. 

This state of affairs gives rise to the following questions. First, is it possible to design \emph{distributed, adaptive, and tractable algorithms jointly optimizing both routing and caching decisions over arbitrary cache network topologies, with provable performance guarantees}? Identifying such algorithms is important precisely due to the combinatorial nature of the problem at hand. Second, presuming such algorithms exist, \emph{do they yield significant performance improvements over fixed routing protocols}? Answering this question in the affirmative may justify the  potential increase in protocol complexity due joint optimization. It can also inform future ICN design, indicating whether full optimization is preferable, or whether an incremental approach in which routing and caching are  separate suffices.

Our goal  is to provide rigorous, comprehensive answers to these two questions. We make the following contributions:
\begin{packeditemize}
\item We show, by constructing a counterexample, that fixed routing (and, in particular, routing towards the nearest server) can be  \emph{arbitrarily suboptimal} compared to jointly optimizing caching and routing decisions. Intuitively, joint optimization affects routing costs drastically because \emph{exploiting path diversity increases caching opportunities}.
\item We propose a formal mathematical framework for joint routing and caching optimization. We consider both \emph{source routing} and \emph{hop-by-hop} routing strategies,  the two predominant classes of routing protocols over the Internet~\citep{kurose2007computer}.
\item We  study the offline version of the joint  routing and caching optimization problem, which is NP-hard, and construct a polynomial-time  $1-1/e$ approximation algorithm. \item  We provide a \emph{distributed}, \emph{adaptive} algorithm that converges to joint routing and caching strategies that are, globally, within a $1-1/e$ approximation ratio from the optimal. \item We evaluate our distributed algorithm over 9 synthetic and 3 real-life network topologies, and show that it significantly outperforms the state of the art: it reduces routing costs  by a factor between 10 and 1000, for a broad array of competitors, including both fixed and dynamic routing protocols.
\end{packeditemize}
The remainder of this paper is organized as follows. We review related work in Section~\ref{sec:related}, and present our mathematical model of a caching network in Section~\ref{sec:model}. Our main results are presented in Section~\ref{sec:main}. A numerical evaluation of our algorithms over several topologies is presented in Section~\ref{sec:evaluation}, and we conclude in Section~\ref{sec:conclusions}.

\section{Related Work}\label{sec:related}

There are several adaptive, distributed approaches determining how to populate caches under fixed routing. A simple, elegant, and ubiquitous algorithm is \emph{path replication}~\citep{cohen2002replication}, sometimes also referred to as ``leave-copy-everywhere'' (LCE) \citep{laoutaris2004meta}: once a request for an item reaches a cache, every downstream node receiving the response caches the item. Several variants of this principle exist. In ``leave-copy-down'' (LCD), a copy is placed \emph{only} in the node immediately preceding the cache storing the requested item \citep{laoutaris2004meta,laoutaris2006lcd}, while ``move-copy-down'' (MCD) also removes the present upstream copy. Probabilistic variants  have also been proposed \citep{psaras2012probabilistic}. To evict items, traditional eviction policies like Least Recently Used (LRU), Least Frequently Used (LFU), First In First Out (FIFO), and Random Replacement (RR) are typically used. Several \change{works} \citep{laoutaris2004meta,psaras2012probabilistic,wang2013optimal,rossini2014coupling,rossi2011caching} have \change{experimentally} studied the performance of these protocols and variants over a broad array of topologies. Despite the advantages of simplicity and elegance inherent in path replication, when targeting an optimization objective such as, e.g., minimizing total routing costs, path replication combined with all of the above eviction and replication policies is known to be arbitrarily suboptimal \citep{ioannidis2016adaptive}.

There is a vast literature on the performance of eviction policies like LRU, FIFO, LFU, etc., on a single cache, and the topic is classic \citep{dan1990approximate,achlioptas2000competitive,thomas1971analysis,flajolet1992birthday,gelenbe1973unified}. Nevertheless, the study of \emph{networks of caches} still poses significant challenges. \techrep{A central problem is that, even if arriving traffic in an LRU cache is, e.g., Poisson, the \emph{outgoing} traffic is hard to describe analytically.}{}  \techrep{Rosensweig et al.~\citep{rosensweig2013steady} study conditions under which path replication with LRU, FIFO, and variants, lead to an ergodic chain. }{}
 A significant breakthrough in this area has been the so-called \emph{Che approximation}  \citep{che2002hierarchical,fricker2012versatile}, which postulates that the hit rate of an LRU cache can be well approximated under the assumption that items stay in the cache for a constant time. This approximation is quite accurate in practice~\citep{fricker2012versatile}, and its success motivated extensive research in so-called \emph{time-to-live} (TTL) caches. \techrep{In TTL caches, items stay in a cache for predetermined times, and are evicted when TTL timers expire.}{} A series of recent works have focused on identifying how to set TTLs to (a) approximate the behavior of known eviction policies, (b) describe hit-rates in closed-form formulas~\citep{fofack2012analysis,che2002hierarchical,berger2014exact,dehghan2016utility,martina2014unified}. Despite these advances, none of the above works address issues of  routing cost minimization over multiple hops, which is our goal.

In their seminal paper~\citep{cohen2002replication} introducing path replication,  Cohen and Shenker also introduced the abstract problem of finding a content placement that minimizes routing costs. The authors show that path replication combined with a constant rate of evictions leads to an allocation that is optimal, in equilibrium, when nodes are visited through uniform sampling.  Unfortunately, this optimality breaks down when uniform sampling is replaced by routing over arbitrary topologies \citep{ioannidis2016adaptive}. Several papers have studied complexity and optimization issues of cost minimization as an offline caching problem under restricted topologies \citep{baev2008approximation,bartal1995competitive,fleischer2006tight,shanmugam2013femtocaching,applegate2010optimal,borst2010distributed}.  With the exception of \citep{shanmugam2013femtocaching}, these works model the network as a bipartite graph:  nodes generating requests connect directly to caches, and demands are satisfied a single hop, and do not readily generalize to arbitrary topologies. In general, the \emph{pipage rounding} technique of Ageev and Sviridenko~\citep{ageev2004pipage} (see also \citep{calinescu2007maximizing,vondrak2008optimal}) yields again a constant approximation algorithm in the bipartite setting, while approximation algorithms are also known for several variants of this problem~\citep{baev2008approximation,bartal1995competitive,fleischer2006tight, borst2010distributed}. Excluding  \citep{borst2010distributed}, all these  works focus only on centralized solutions of the offline caching problem; none considers jointly optimizing caching \emph{and} routing decisions.

 In earlier work~\citep{ioannidis2016adaptive}, we consider a  setting in which routes are fixed, and only caching decisions are optimized in an adaptive, distributed fashion.  We extend \citep{ioannidis2016adaptive} to incorporate routing decisions, both through source and hop-by-hop routing. We show that a variant of pipage rounding \citep{ageev2004pipage} can be used to construct a poly-time approximation algorithm, that also lends itself to a distributed, adaptive implementation. Crucially, our evaluations in Section~\ref{sec:evaluation} show  that jointly optimizing caching \emph{and} routing  significantly improves performance compared to fixed routing, reducing the routing costs    by as much as three orders of magnitude compared to \citep{ioannidis2016adaptive}.

 Several recent works study caching and routing \emph{jointly}, in more restrictive settings than the ones we consider here. The benefit of routing towards nearest replicas, rather than towards nearest designated servers, has been observed  empirically \citep{chiocchetti2012exploit,fayazbakhsh2013less,carofiglio2016joint}.  Deghan et al.~\citep{dehghan2014complexity}, Abedini and Shakkotai \citep{abedini2014content}, \change{and Xie et al.~\citep{xie2012tecc}} all study  joint routing and content placement schemes in a bipartite, single-hop setting.  \change{In all three cases, } minimizing the single-hop routing cost reduces to  solving a linear program;  Naveen et al.~\citep{naveen2015interaction}  extend this to other, non-linear (but still convex) objectives of the hit rate, still under single-hop, bipartite routing constraints. \change{None of these approaches  generalize to a multi-hop setting, which leads to non-convex formulations (see~Section~\ref{sec:problemformulation}); addressing this lack of convexity is one of our technical contributions.} 
Closer to our work, a multi-hop, multi-path setting is formally analyzed by Carofiglio et al.~\citep{carofiglio2016joint} under the assumption that requests by different users follow non-overlapping paths. The authors show that under appropriate conditions on request arrival rates, this assumption leads to a convex optimization problem. Our approach  addresses the lack of convexity in its full generality, for arbitrary topologies, request arrival rates, and overlapping  paths.

\change{The problem we study is also related to more general placement problems, including the allocation of virtual machines (VMs) to hosts in  cloud computing \citep{li2011virtual,guenter2011managing,van2010cost,batista2007set}--see also \citep{jiang2012joint}, that jointly optimizes placement and routing in this context. This is a harder problem: heterogeneity of host resources and VM requirements leads to multiple knapsack-like constraints (one for each resource) per host. Our storage constraints are simpler; as a result, in contrast to \citep{li2011virtual,guenter2011managing,van2010cost,batista2007set,jiang2012joint}, we can provide poly-time, distributed algorithms with provable approximation guarantees.}

 \section{Model}\label{sec:model}
We begin by presenting our formal model, extending \citep{ioannidis2016adaptive} to account for both caching \emph{and} routing decisions. Our analysis applies to two routing variants: (a) \emph{source routing} and (b) \emph{hop-by-hop routing}.  \change{In both cases, we study two types of strategies: \emph{deterministic} and \emph{randomized}. For example, in source routing,  requests for an item originating from the same source may be forwarded over several possible paths, given as input. In deterministic source routing, \emph{only one} is selected and used \emph{for all subsequent requests} with this origin. In contrast, a randomized strategy samples a new path to follow independently with each new request. We also use similar deterministic and randomized analogues both for caching strategies as well as for hop-by-hop routing strategies.}

\change{Randomized strategies subsume deterministic ones, and are arguably more flexible and general. This begs the question: why  study both?  There are three reasons. 
First, optimizing deterministic strategies naturally relates to submodular maximization subject to matroid constraints,  allowing us to leverage related combinatorial optimization techniques. Second,
the online, distributed algorithms we propose to construct randomized strategies rely on the solution to the offline, deterministic problem. Finally, and most importantly: deterministic strategies turn out to be equivalent to randomized strategies! As we show in Thm.~\ref{cor:equiv},  the smallest routing cost attained by randomized strategies is exactly the same as the one attained by deterministic strategies. }

\begin{table}[!t]
\begin{footnotesize}
\begin{tabular}{p{0.1\textwidth}p{0.8\textwidth}}
\hline
\multicolumn{2}{c}{\textbf{Common Notation}}\\
\hline
$G(V,E)$ & Network graph, with nodes $V$ and edges $E$\\
$\catalog$ & Item catalog\\
$c_v$ & Cache capacity at node $v\in V$\\
$w_{uv}$ & Weight of edge $(u,v)\in E$\\
$\requests$ & Set of requests $(i,s)$, with $i\in\catalog$ and source $s\in V$\\
$\lambda_{(i,s)}$ & Arrival rate of requests $(i,s)\in \requests$ \\
$\source_i$ & Set of designated servers of $i\in \catalog$\\
$x_{vi}$ & Variable indicating whether $v\in V$ stores $i\in\catalog$ \\
$\xi_{vi}$ & Marginal probability that $v$ stores $i$\\ 
$X$ & Global caching strategy  of $x_{vi}$s, in $\{0,1\}^{|V|\times |\catalog|}$\\%/relaxed variable $x_{vi}$\\
$\Xi$ & Expectation of caching strategy matrix $X$\\
$T$ & Duration of a timeslot in online setting\\
$w_{uv}$ & weight/cost  of edge $(u,v)$\\
$\supp(\cdot)$ & Support of a probability distribution\\
$\conv(\cdot)$ & Convex hull of a set\\
\hline
\multicolumn{2}{c}{\textbf{Source Routing}}\\
\hline
$\pathset_{(i,s)}$ & Set of paths request $(i,s)\in \requests$ can follow\\
$P_\SR$ & Total number of paths\\
$p$ & A simple path of $G$\\
$k_{p}(v)$ & The position of node $v\in p$ in path $p$.\\ 
$r_{(i,s),p}$ & Variable indicating  whether $(i,s)\in\requests$ is forwarded over $p\in\pathset_{(i,s)}$ \\
$\rho_{(i,s),p}$ & Marginal probability that $s$ routes request for $i$ over $p$\\ $r$ & Routing strategy of $r_{(i,s),p}$s,  in $\{0,1\}^{\sum_{(i,s)\in \requests}|\pathset_{(i,s)}|}$.\\
$\rho$ & Expectation of routing strategy vector $r$\\
$\feasibledomain_\SR$ &Feasible strategies $(r,X)$ of \CGS\\
RNS & Route to nearest server\\
RNR & Route to nearest replica\\
\hline
\multicolumn{2}{c}{\textbf{Hop-by-Hop Routing}}\\
\hline
$G^{(i)}$& DAG with sinks in $\source_i$\\
$E^{(i)}$ & Edges in DAG $G^{(i)}$\\
$G^{(i,s)}$ & Subgraph of $G^{(i)}$ including only nodes reachable from $s$\\
 $\pathset_{(i,s)}^u$ & Set of  paths in $G^{(i,s)}$ from $s$ to $u$. \\
$P_\HH$ & Total number of paths\\
$r_{uv}^{(i)}$ &   Variable indicating whether $u$ forwards a request for $i$ to $v$\\
$\rho_{uv}^{(i)}$ & Marginal probability that $u$ forwards a request for $i$ to $v$\\
$r$ & Routing strategy of $r_{u,v}^i$s,  in $\{0,1\}^{\sum_{i\in \catalog}|E^{(i)}|}$.\\
$\rho$ & Expectation of routing strategy vector $r$\\
$\feasibledomain_\HH$ & Feasible strategies $(r,X)$ of  \CGHH\\
\hline
\end{tabular}
\end{footnotesize}
\caption{Notation Summary}
\end{table}

\subsection{Network Model and Content Requests}

Consider a network represented as a directed, symmetric\footnote{A directed graph is symmetric when $(i,j)\in E$ implies that $(j,i)\in E$.} graph $G(V,E)$.  
Content items (e.g., files, or file chunks) of equal size   are to be distributed across network nodes. Each node is associated with a cache that can store a finite number of items. We denote by  $\catalog$ the set of possible content items, i.e., the \emph{catalog}, and by $\capacity_v\in \naturals$ the cache capacity at node $v\in V$: exactly $\capacity_v$ content items can be stored in $v$. 
The network serves content requests routed over the graph $G$.  A request $(i,s)$ is determined by (a) the  item $i\in \catalog$ requested, and (b) the source $s\in V$ of the request. 
 We denote by $\requests\subseteq \catalog\times V$  the set  of all requests. Requests of different types $(i,s)\in\requests$ arrive according to independent Poisson processes with arrival rates $\lambda_{(i,s)}>0$, $(i,s)\in \requests$. 

For each item $i\in \catalog$  there is a fixed set of \emph{designated server} nodes $\source_i\subseteq V$, that always store $i$. A node $v\in \source_i$ permanently stores $i$ in \emph{excess memory outside its cache}. Thus, the placement of items to designated servers is fixed and outside the network's design. 
A request $(i,s)$ is routed over a path in $G$ towards a designated server. However,  forwarding terminates upon reaching \emph{any intermediate cache} that stores $i$. At that point, a response carrying $i$ is sent over  the reverse path, i.e., from the node where the cache hit occurred, back to source node $s$.   Both caching \emph{and} routing decisions are network design parameters, which we define formally below.

\subsection{Caching Strategies}
\change{We study two types or caches: \emph{deterministic} and \emph{randomized}.}

\noindent\change{\textbf{Deterministic caches.}}
For each node $v\in V$, we define  \emph{$v$'s caching strategy} as a vector $x_v\in\{0,1\}^{|\catalog|}$, where $x_{vi}  \in \{0,1\},$  for $ i \in \catalog $,
is the binary variable indicating whether $v$ stores content item $i$. As $v$ can store no more than $c_v$ items, we have that: 
\begin{align}\textstyle\sum_{i\in \catalog} x_{vi} \leq c_v,\text{ for all }v\in V.\label{capconst}\end{align}
 We define the  \emph{global caching strategy} as the matrix
 $\textstyle X=[x_{vi}]_{v\in V,i \in \catalog}\in \{0,1\}^{|V|\times |\catalog|},$ 
whose rows comprise the caching strategies of each node.

\noindent\change{\textbf{Randomized caches}. In the case of randomized caches}, the caching strategies $x_v$, $v\in V$, are random variables. We denote by:
\begin{align}\xi_{vi} \equiv \prob[x_{vi} = 1] = \expect[x_{v,i}]\in [0,1], \text{ for }i\in C,\end{align}
the marginal probability that node $v$ caches item $i$, and by 
$\Xi=[\xi_{vi}]_{v\in V,i\in \catalog}=\expect[X]\in [0,1]^{|V|\times |\catalog|},$
the corresponding expectation of the global caching strategy.

\begin{figure}[!t]
\includegraphics[width=\columnwidth]{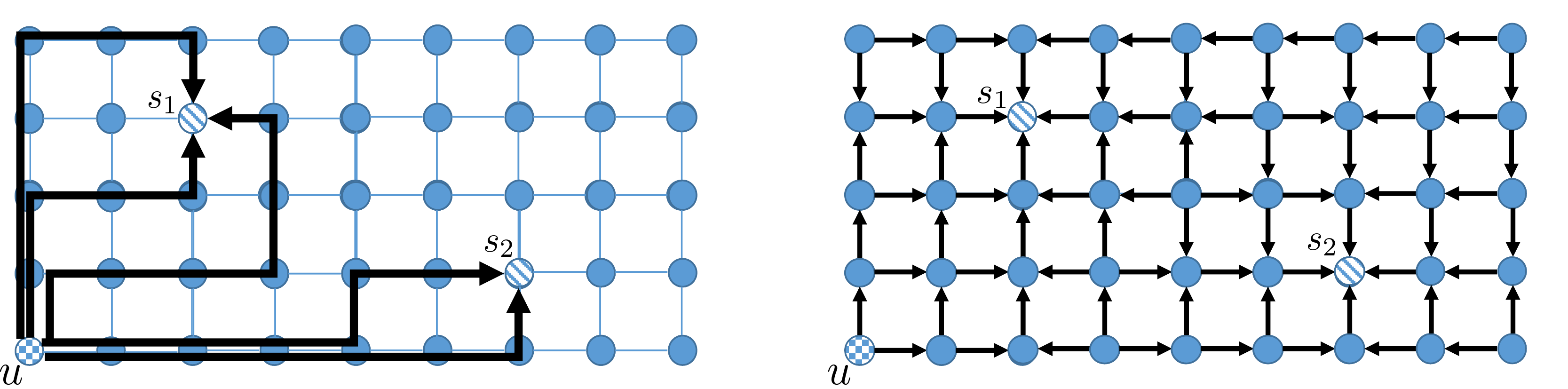}
\caption{Source Routing vs.~Hop-by-Hop routing. In source routing, shown left, source node $u$ on the bottom left can choose among 5 possible paths to route a request to one of the designated servers storing $i$ ($s_1,s_2$). In hop-by-hop routing, each intermediate node in the network selects the next hop among one of its neighbors in a DAG, whose sinks are the designated servers.}
\end{figure}
\subsection{Source Routing Strategies}

Recall that requests are routed towards designated server nodes. In source routing, for every request $(i,s)\in \catalog\times V$, there exists a $\pathset_{(i,s)}$ of \emph{paths} that the request can follow towards a designated server in $\source_i$. A source node $s$ can forward a request among any of these paths, but we assume each response  follows the same path as its corresponding request.

Formally,  a path $\ppath$ of length $|p|=K$ is a sequence $\{p_1,p_2,\ldots,p_K\}$ of nodes  $p_k\in V$ such that  $(p_k,p_{k+1})\in E$, for every $k\in \{1,\ldots,|p|-1\}$.  
We make the following natural  assumptions on the set of paths $\pathset_{(i,s)}$. For every $\ppath\in\pathset_{(i,s)}$:
(a) $p$ starts at $s$, i.e., $p_1=s$;
(b) $p$  is simple, i.e., it contains no loops;
(c) the last node in $p$  is a designated server for item $i$, i.e., if $|p|=K$, 
$p_K\in \source_i$; and
(d) no other node in $p$ is a designated server for $i$, i.e., if $|p|=K$,
$p_k\notin \source_i$, for $k=1,\ldots,K-1.$ 
Given a path $p$ and a $v\in p$, denote by  
$k_p(v)$ is the position of $v$ in $p$; i.e., $k_p(v)$ equals to $k\in \{1,\ldots,|p|\}$ such that $p_k=v$. 
\change{As in the case of caches, we consider both deterministic and randomized routing strategies.}

\sloppy
\noindent\change{\textbf{Deterministic Routing.}}
Given sets $\pathset_{(i,s)}$, $(i,s)\in \requests$, the \emph{routing strategy} of a source $s\in V$ w.r.t.~request $(i,s)\in \requests$ is a vector $r_{(i,s)}\in \{0,1\}^{|\pathset_{(i,s)}|}$, where
 $r_{(i,s),p}\in \{0,1\}$   is a binary variable indicating whether $s$ selected path $p \in \pathset_{(i,s)}$.
Routing strategies satisfy:  \begin{align}\textstyle\sum_{p\in \pathset_{(i,s)}} r_{(i,s),p} =1,\text{  for all }(i,s)\in\requests.\label{routecap}\end{align}
indicating that exactly one path is selected.  Let $P_\SR =\sum_{(i,s)\in \requests}|\pathset_{(i,s)}|$
be the total number of paths.
We refer to the vector $\textstyle r= [ r_{(i,s),p}]_{(i,s)\in \requests,p\in \pathset_{(i,s)}}\in  \{0,1\}^P,$ as the  \emph{global routing strategy}. 

\noindent\change{\textbf{Randomized Routing.}} In the case of randomized routing, variables $r_{(i,s)}$, $(i,s)\in \requests$ are random. We randomize routing by allowing requests to be routed over a random path in $\pathset_{(i,s)}$, selected  \emph{independently of all past requests (at $s$ or elsewhere)}.  We denote by 
\begin{align}\rho_{(i,s),p}  \equiv \prob[r_{(i,s),p}=1 ] = \expect[r_{(i,s),p}] , \text{ for }p\in\pathset_{(i,s)},\label{marginalroute}\end{align}
the probability that path $p$ is selected by $s$, and by
$\rho=[\rho_{(i,s),p}]_{(i,s)\in\requests ,p\in \pathset_{(i,s)} }=\expect[r]\in [0,1]^{P} $
 the expectation of the global routing strategy $r$.

\fussy

\noindent\change{\textbf{Remark.}} We make no a priori assumptions on $P_\SR$, the total number of  paths used during source routing; moreover, we allow paths  to overlap. The complexity of our offline algorithm, and the rate of convergence of our distributed, adaptive algorithm depend on $P_\SR$. In practice, if the number of possible paths is, e.g., exponential in $|V|$, it makes sense to restrict  each $\pathset_{(i,s)}$ to a small subset of possible paths, or to use hop-by-hop routing instead, which, as discussed below,  restricts the maximum number of paths considered.

\subsection{Hop-by-Hop Routing Strategies}
Under hop-by-hop routing, each node along the path makes an individual decision on where to route a request message. When a request for item $i$ arrives at an intermediate node $v\in V$, node $v$ determines how to forward the request to one of its neighbors. The decision depends on $i$ but \emph{not} on the request's source. This limits the paths a request may follow, making hop-by-hop routing less expressive than source routing. On the other hand, reducing the space of routing strategies reduces complexity. In adaptive algorithms, it also speeds up convergence, as routing decisions  w.r.t.~$i$  are ``learned'' across requests by different sources.

To ensure loop-freedom, we must assume that forwarding decisions are restricted to a subset of possible neighbors in $G$.  For each $i\in \catalog$, we denote by $G^{(i)}(V,E^{(i)})$ a graph that has the following properties:
(a) $G^{(i)}$ is a subgraph of $G$, i.e., $E^{(i)}\subseteq E$;
(b) $G^{(i)}$ is a directed acyclic graph (DAG); and
(c) a node $v$ in $G^{(i)}$ is a sink if and only if it is a designated server for $i$, i.e., $v\in \source_i$. 
Note that, given $G$ and $\source_i$,  $G^{(i)}$ can be constructed in polynomial time using, e.g.,  the Bellman-Ford algorithm \citep{cormen2009book}. Indeed,   requiring that $v$ forwards requests for $i\in \catalog$ only towards neighbors with a smaller distance to a designated server in $\source_i$ results in such a DAG. A distance-vector protocol~\citep{kurose2007computer} can form this DAG in a distributed fashion. We assume that every node $v\in V$ can forward a request for item $i$ \emph{only to a neighbor in $G^{(i)}$}. Then, the above properties of $G^{(i)}$ ensure both loop freedom and successful termination.

\sloppy \noindent\change{\textbf{Deterministic Routing.}} For any node $s\in V$, let $G^{(i,s)}$ be the induced subgraph of $G^{(i)}$ which results from removing any nodes in $G^{(i)}$  not reachable from $s$. For any $u$ in $G^{(i,s)}$, let $\pathset_{(i,s)}^u$ be the set of all paths in $G^{(i,s)}$ from $s$ to $u$, and denote by $P_\HH= \sum_{(i,s)\in \catalog} \sum_{u\in V} |\pathset_{(i,s)}^u|$.
  We denote by $r_{uv}^{(i)}\in \{0,1\},$  for $(u,v)\in E^{(i)}$, $i\in \catalog,$ the decision variable indicating whether $u$ forwards a request for $i$ to $v$. The \emph{global routing strategy} is $r = [r_{uv}^{(i)}]_{i\in\catalog,(u,v)\in E^{(i)}}\in \{0,1\}^{\sum_{i\in\catalog}|E^{(i)}|},$ and  satisfies 
\begin{align}\textstyle\sum_{v:(u,v)\in E^{(i)}}r_{uv}^{(i)} =1,\quad \text{for all}~v\in V, i\in\catalog. \end{align}
 Note that, in contrast to source routing strategies, that have length $P_\SR$,  hop-by-hop routing strategies have length at most $|\catalog||E|$.
 
 \fussy
\noindent\change{\textbf{Randomized Routing.}} As in the case of source routing, we  also consider randomized hop-by-hop routing strategies, whereby each request is forwarded independently from previous routing decisions to one of the possible neighbors. We again denote by 
\begin{align}\begin{split}\rho &= [\rho_{uv}^{(i)}]_{i\in\catalog,(u,v)\in E^{(i)}}=[\expect[r_{uv}^{(i)}]]_{i\in\catalog,(u,v)\in E^{(i)}}\\&= \big[\prob[r_{uv}^{(i)}=1]\big]_{i\in\catalog,(u,v)\in E^{(i)}} \in [0,1]^{\sum_{i\in\catalog}|E^{(i)}|},\end{split}\end{align}
the vector of corresponding (marginal) probabilities of routing decisions at each node $v$. 

\subsection{Offline vs. Online Setting}\label{sec:offon}

To reason about the caching networks we have proposed, we consider two settings: the \emph{offline} and \emph{online} setting. In the offline setting, all problem inputs (demands, network topology, cache capacities, etc.) are known apriori to, e.g., a system designer. At time $t=0$, the system designer selects (a) a caching strategy $X$, and (b) a routing strategy $r$. Both can be either deterministic or randomized, but both are also static: they do not change as time progresses. In the case of caching, cache contents (selected deterministically or at random at $t=0$) remain static for all $t\geq 0$. In the case of routing decisions, the distribution over paths (in source routing) or neighbors (in hop-by-hop routing) remains static, but each request is routed independently of previous requests.

In the online setting,  no a priori knowledge of the demand, i.e., the rates of requests $\lambda_{(i,s)}$, $(i,s)\in \requests$  is assumed. Both caching \emph{and} routing strategies change through time via a \emph{distributed, adaptive} algorithm. Time is slotted, and  each slot has duration $T>0$. During  a timeslot, both caching and and routing strategies remain fixed. Nodes have access only to local information: they are aware of their graph neighborhood and state information they maintain locally. They exchange messages, including both normal request and response traffic, as well as (possibly) control messages, and may adapt their state. At the conclusion of a time slot, each node changes its caching and routing strategies. Changes made by $v$ depend only on its neighborhood, its current local state, as well as on messages that node $v$ received in the previous timeslot. Both caching and routing strategies during a timeslot may be deterministic or randomized. Implementing a caching strategy at the conclusion of a timeslot involves changing cache contents, which incurs additional overhead; if $T$ is large, however, this cost is negligible compared to the cost of transferring items during a timeslot.

\subsection{Optimal Routing and Caching}\label{sec:problemformulation}
We are now ready to formally pose the problem of jointly optimizing caching and routing. We pose here the \emph{offline} problem, in which problem inputs are given, and static caching and routing strategies are determined (jointly) at time $t=0$. Nonetheless, we will devise distributed, adaptive algorithms that do not a priori know the demand, but still converge to (probabilistic) strategies that are within a constant approximation of the (offline) optimal.

 To capture  costs  (e.g., latency, money, etc.), we associate a \emph{weight} $w_{uv}\geq 0$ with each  edge $(u,v)\in E$, representing the cost of transferring an item across this edge.
 We assume  that costs are solely due to response messages that carry an item, while request forwarding costs are negligible.  We do not assume that $w_{uv}=w_{vu}$. We describe the cost minimization objectives under source and hop-by-hop routing below.

\noindent \textbf{Source Routing.} The cost for serving a request $(i,s)\in \requests$ under source routing is:
\begin{align} C^{(i,s)}_{\SR}(r,X) =\!\! \sum_{p\in \pathset_{(i,s)}} \!\!\! r_{(i,s),p} \!\sum_{k=1}^{|p|-1}\!\! w_{p_{k+1}p_k}\!\!\prod_{k'=1}^k (1\!-\!x_{p_{k'}i}).\label{sourcecost}\end{align}
Intuitively, \eqref{sourcecost} states that $C^{(i,s)}_{\SR}$ includes the cost of an edge $(p_{k+1},p_k)$ in the path $p$ if (a) $p$ is selected by the routing strategy, and (b) \emph{no} cache preceding this edge in $p$  stores $i$. 
In the deterministic setting, we seek a  global caching and routing strategy $(r,X)$ minimizing the aggregate expected cost, defined as:
\begin{align}\textstyle C_{\SR}(r,X) = \sum_{(i,s)\in \requests} \lambda_{(i,s)}  C^{(i,s)}_{\SR}(r,X),\label{routingcost}\end{align}
with $C^{(i,s)}_\SR$ given by \eqref{sourcecost}.
That is, we wish to solve:

 \begin{subequations}\label{deterministicsourcecost}
\center{\textsc{MinCost}-\SR}\vspace*{-0.5em}
\begin{align}
\text{Minimize:}& \quad C_{\SR}(r,X)\\
\text{subj.~to:}& \quad (r,X)\in \feasibledomain_\SR
\end{align}
\end{subequations}
where 
$\feasibledomain_\SR\subset \reals^{P_\SR} \times \reals^{|V|\times |\catalog|}$  is the set of  $(r,X)$ satisfying the routing, capacity, and integrality constraints, i.e.:
\begin{subequations}\label{integralconstr-sr}
\begin{align}
& \textstyle\sum_{i\in\catalog} x_{vi} = \capacity_v, & &\text{ for all }v\in V, \label{xcapacity}\\
& \textstyle\sum_{p\in \pathset_{(i,s)}} r_{(i,s),p} =1, &&\text{ for all }(i,s)\in\requests, \label{pselect}\\
&x_{vi}\in \{0,1\}, &&\text{ for all }v\in V, i\in\catalog, \text{ and} \label{xintegrality}\\
&r_{(i,s),p}\in \{0,1\}, &&\text{ for all }p\in\pathset_{(i,s)},(i,s)\in \requests.  \label{pintegrality}
\end{align}
\end{subequations}
This problem is NP-hard, even in the case where routing is fixed: see Shanmugam et al.~\citep{shanmugam2013femtocaching} for a reduction from the 2-Disjoint Set Cover Problem.

\noindent \textbf{Hop-By-Hop Routing.} Similarly to \eqref{sourcecost}, under hop-by-hop routing, the cost of serving $(i,s)$ can be written as:
\begin{align}\begin{split}
C^{(i,s)}_{\HH}(r,X) = \textstyle\sum_{(u,v)\in G^{(i,s)}} w_{vu} \cdot r_{uv}^{(i)} (1-x_{ui}) \cdot\\
\textstyle\sum_{p\in \pathset_{(i,s)}^u} \prod_{k'=1}^{|p|-1} r_{p_{k'}p_{k'+1}}^{(i)}(1-x_{p_{k'}i}).\end{split}\label{hopcost}\end{align}
We wish to solve:
\begin{subequations}\label{deterministichopbyhop}
\center{\textsc{MinCost}-\HH}\vspace*{-0.5em}
\begin{align}
\text{Minimize:}& \quad C_{\HH}(r,X)\\
\text{subj.~to:}& \quad (r,X)\in \feasibledomain_\HH
\end{align}
\end{subequations}
where 
$C_{\HH}(r,X) = \textstyle\sum_{(i,s) \in \requests }\lambda_{(i,s)}  C^{(i,s)}_{\HH}(r,X)$ 
is the expected routing cost, and
$\feasibledomain_\HH$ is the set of $(r,X)\in \reals^{\sum_{i\in C}|E^{(i)}|} \times \reals^{|V|\times |\catalog|}$ satisfying the  constraints:
\begin{subequations}\label{integralconstr-hbh}
\begin{align}
& \textstyle\sum_{i\in\catalog} x_{vi} = \capacity_v, & &\text{ for all }v\in V, \label{xhcapacity}\\
&\textstyle\sum_{v:(u,v)\in E^{(i)}}r_{uv}^{(i)} =1   &&\text{ for all }v\in V,i\in\catalog, \label{phselect}\\
&x_{vi}\in \{0,1\}, &&\text{ for all }v\in V, i\in\catalog, \text{ and} \label{xhintegrality}\\
&r_{uv}^{(i)}\in \{0,1\}, &&\text{ for all }(u,v)\in E^{(i)}, i\in \catalog.  \label{phintegrality}
\end{align}
\end{subequations}

\noindent\textbf{Randomization.} The above routing cost minimization problems can also be stated in the context of randomized caching and routing strategies. 
For example, in the case of source routing, assuming\footnote{The independence of routing and caching strategies across nodes comes \emph{without loss of generality}! The minimum expected cost attainable under independence is the same under cache contents and routing strategies sampled from \emph{an arbitrary joint distribution over $\mathcal{D}_\SR$}; see Thm~\ref{cor:equiv}.} 
(a) independent caching strategies across nodes selected at time $t=0$, with marginal probabilities given by $\Xi$, and (b)  independent routing strategies at each source, with marginals given by $\rho$ (also independent from caching strategies), all terms in $C_\SR$ contain products of independent random variables; this implies that:
\begin{align}\expect[C_{\SR}(r,X)] = C_{\SR}[\expect[r],\expect[X]]=C_{\SR}(\rho,\Xi),\label{relaxfun}\end{align}
where the expectation  is taken over the randomness of both caching and routing strategies. The expected routing cost thus depends on the routing and caching strategies only through the expectations $\rho$ and $\Xi$. As a result,  under randomized routing and caching strategies, \textsc{MinCost}-\SR{} becomes (see  \techrep{Appendix~\ref{app:randomisrelaxed}}{\citep{arxivthis}} for the derivation): \begin{subequations}\label{relaxedmin}
\begin{align} \text{Minimize:} &\quad  C_\SR(\rho,\Xi)\\
\text{subj.~to:} &\quad (\rho,\Xi)\in \conv(\feasibledomain_\SR)
\end{align}
\end{subequations}
where 
$\conv(\feasibledomain_\SR)$ is the convex hull of $\feasibledomain_\SR$; this is precisely the set defined by \eqref{integralconstr-sr} with integrality constraints \eqref{xintegrality}, \eqref{pintegrality} relaxed. The objective function $C_\SR$ is \emph{not convex} and the relaxed problem \eqref{relaxedmin} is therefore \emph{not} a convex optimization problem. This is in stark contrast to single-hop settings, that often can naturally be expressed as linear programs \citep{dehghan2014complexity,abedini2014content,naveen2015interaction}.

A similar derivation can be done for hop-by-hop routing.  Assuming again independent caches and independent routing strategies, it can be shown that optimizing over randomized hop-by-hop strategies is equivalent to
\begin{subequations}\label{relaxedminHH}
\begin{align} \text{Minimize:} &\quad  C_\HH(\rho,\Xi)\\
\text{subj.~to:} &\quad (\rho,\Xi)\in \conv(\feasibledomain_\HH),
\end{align}
\end{subequations}
where $\conv(\feasibledomain_\HH)$  the convex hull of $\feasibledomain_\HH$. This, again, is a non-convex optimization problem.

\subsection{Fixed Routing}\label{sec:fixedrouting}

When the global routing strategy $r$ is fixed,  \eqref{deterministicsourcecost} reduces to
\begin{subequations}\label{fixedrouting}
\begin{align}
\text{Minimize:}& \quad C_{\SR}(r,X)\\
\text{subj.~to:}& \quad X \text{ satisfies } \eqref{xcapacity}\text{ and }\eqref{xintegrality} 
\end{align}
\end{subequations}
\textsc{MinCost}-\HH{} can be similarly restricted to caching only. We studied this restricted optimization in earlier work \citep{ioannidis2016adaptive}. In particular, under given global routing strategy $r$, we cast \eqref{fixedrouting} as a maximization problem as follows. Let 
\begin{align}C_0^r = C_\SR(r,0) =  \sum_{(i,s)\in \requests} \lambda_{(i,s)} \sum_{p\in \pathset_{(i,s)}} \!\!\! r_{(i,s),p} \!\sum_{k=1}^{|p|-1}\!\! w_{p_{k+1}p_k} \end{align}
be the cost when all caches are empty (i.e., $X$ is the zero matrix 0). Note that this is a constant that does not depend on $X$.
 Consider the following maximization problem:
\begin{subequations}\label{maxprobfixr}
\begin{align}
\text{Maximize:}& \quad F^r_{\SR}(X) =  C_0^r-C_\SR(r,X) \\
\text{subj.~to:}& \quad X \text{ satisfies } \eqref{xcapacity}\text{ and }\eqref{xintegrality} 
\end{align}
\end{subequations}
This problem is equivalent to \eqref{fixedrouting}, in that a feasible solution to \eqref{maxprobfixr} is optimal if and only if it also optimal for \eqref{fixedrouting}. The objective $F^r_{\SR}(X)$, referred to as the \emph{caching gain} in \citep{ioannidis2016adaptive}, is  monotone, non-negative, and  submodular, while the set of constraints on $X$ is a set of matroid constraints. As a result, for any $r$, there exist standard approaches for constructing a polynomial time approximation algorithm solving the corresponding maximization problem \eqref{maxprobfixr} within a $1-1/e$ factor from its optimal solution \citep{ioannidis2016adaptive,shanmugam2013femtocaching}. In addition, we show  \citep{ioannidis2016adaptive} that an approximation algorithm based on a technique known as \emph{pipage rounding}~\citep{ageev2004pipage} can be converted into a distributed, adaptive version with the same approximation ratio.

\subsection{Greedy Routing Strategies}\label{sec:greedyrouting}

In the case of source routing, we identify two ``greedy'' deterministic routing strategies, that are often used in practice, and play a role in our analysis.
We say that a global routing strategy $r$ is a \emph{route-to-nearest-server} (RNS) strategy if all paths it selects are least-cost paths to designated servers, irrespectively of cache contents. Formally, for all $(i,s)\in \requests$, 
$\textstyle r_{(i,s),p^*}=1$ for some $p^*\in \argmin_{p\in\pathset_{(i,s)}} \sum_{k=1}^{|p|-1}w_{p_{k+1},p_{k}},$
while $r_{(i,s),p}=0$ for all other $p\in \pathset_{(i,s)}$ s.t. $p\neq p^*$.
Similarly, given a caching strategy $X$, we say that a global routing strategy $r$ is \emph{route-to-nearest-replica} (RNR) strategy if, for all $(i,s)\in \requests$, 
$r_{(i,s),p^*}=1$ for some $p^*\in \argmin_{p\in\pathset_{(i,s)}} \sum_{k=1}^{|p|-1}w_{p_{k+1},p_{k}}\!\!\prod_{k'=1}^k (1\!-\!x_{p_{k'}i}),$ 
while $r_{(i,s),p}=0$ for all other $p\in \pathset_{(i,s)}$ s.t. $p\neq p^*$.
In contrast to RNS strategies, RNR strategies  depend on the caching strategy $X$. Note that RNS and RNR strategies can be defined similarly in the context of hop-by-hop routing.

\techrep{
\subsection{Multi-Item Requests and Stationarity}
We have assumed that all items in the catalog are of \emph{equal size}. In practice, contents of unequal size can be partitioned into equally sized chunks, which would play the role of ``items'' in our formulation.  Requesting, e.g., a file, amounts to requesting all of its chunks simultaneously. Our model can be easily extended to account for requests comprising multiple item/source tuples, i.e., of the form $\{(i_1,s),(i_2,s),\ldots,(i_L,s)\}$. Each such multi-item request contributes several summation terms to the routing cost $C_\SR$, given by \eqref{routingcost}, one for each chunk in the requested set. Thus, the objective has the same form. Note that this formulation does not assume, or require, that chunks of the same file are stored in the same cache or that requests for them follow the same paths. 

More broadly speaking, the time average routing cost is given by $CS_{\SR}$ even if (a) the request arrival processes for requests of type $(i,s)\in \requests$ \emph{are not independent}, and (b)  are arbitrary stationary ergodic renewal processes with arrival rates $\lambda_{(i,s)}$. Our offline analysis and results therefore readily extend to general stationary ergodic arrival processes. Moreover our distributed, adaptive algorithm readily applies to non-independent Poisson arrivals (and in particular, to multi-item requests that would occur under file-chunking); it can also be extended to general stationary arrivals, provided that the subgradient estimators used in our projected gradient ascent algorithm are replaced with unbiased estimates matching the statistics of underlying renewal process.  
}{}

\section{Main Results}\label{sec:main}
We present our main results in this section, extending the analysis in \citep{ioannidis2016adaptive} to the joint optimization of both caching and routing decisions. We provide an analysis of both source and hop-by-hop routing; \techrep{all proofs are provided in the appendix.}{proofs of theorems are omitted, and are provided in our technical report \citep{arxivthis}.}

\subsection{Routing to Nearest Server Is Suboptimal}\label{sec:subopt}
 A simple approach, followed by most works that optimize caching separately from routing, is to always route requests to the nearest designated server storing an item (i.e., use an RNS strategy).  It is therefore interesting to ask how this simple heuristic performs compared to a solution that attempts to solve \eqref{deterministicsourcecost} by \emph{jointly} optimizing caching and routing.  It is easy to see that RNS and, more generally, routing that ignores caching strategies, can lead to arbitrarily suboptimal solutions:  
\begin{theorem}\label{thm:subopt} For any $M>0$, there exists a caching network for which the route-to-nearest-server strategy $r'$ satisfies
\begin{align}{ \min_{X:(r',X)\in \feasibledomain_\SR} C_\SR(r',X)  }/{\min_{(r,X)\in \feasibledomain_\SR}C_\SR(r,X)} = \Theta(M). \end{align}
\end{theorem}
In other words, routing to the nearest server can be arbitrarily suboptimal, incurring a cost arbitrarily larger than the cost of the optimal jointly optimized routing and caching policy. The network that exhibits this behavior is shown in Fig.~\ref{fig:simpleexample}, and a proof of the theorem can be found in \techrep{Appendix~\ref{app:proofofthmsubopt}.}{\citep{arxivthis}}.  \change{In short, a source node $s$ generates requests for items $1$ and $2$ that are permanently stored on designated server $t$. There are two alternative paths towards $t$, each passing through an intermediate node with cache capacity 1 (i.e., able to store only one item).  Under shortest path routing, requests for both items are forwarded over the path of length $M+1$  towards $t$; fixing routes this way leads to a cost of $M+1$ for at least one of the items, irrespectively of which item is cached in the intermediate node. On the other hand, if routing and caching decisions are jointly optimized, requests for the two items can be forwarded to different paths, allowing both items to be cached, and reducing the cost for both requests to at most 2. }

This example illustrates that joint optimization of caching \emph{and} routing decisions benefits the system by \emph{increasing path diversity}. In turn, increasing path diversity can increase caching opportunities, thereby leading to reductions in caching costs. This is consistent with our experimental results in Section~\ref{sec:evaluation}. 

 \begin{figure}[!t]
 \begin{center}
 \includegraphics[width=0.7\columnwidth]{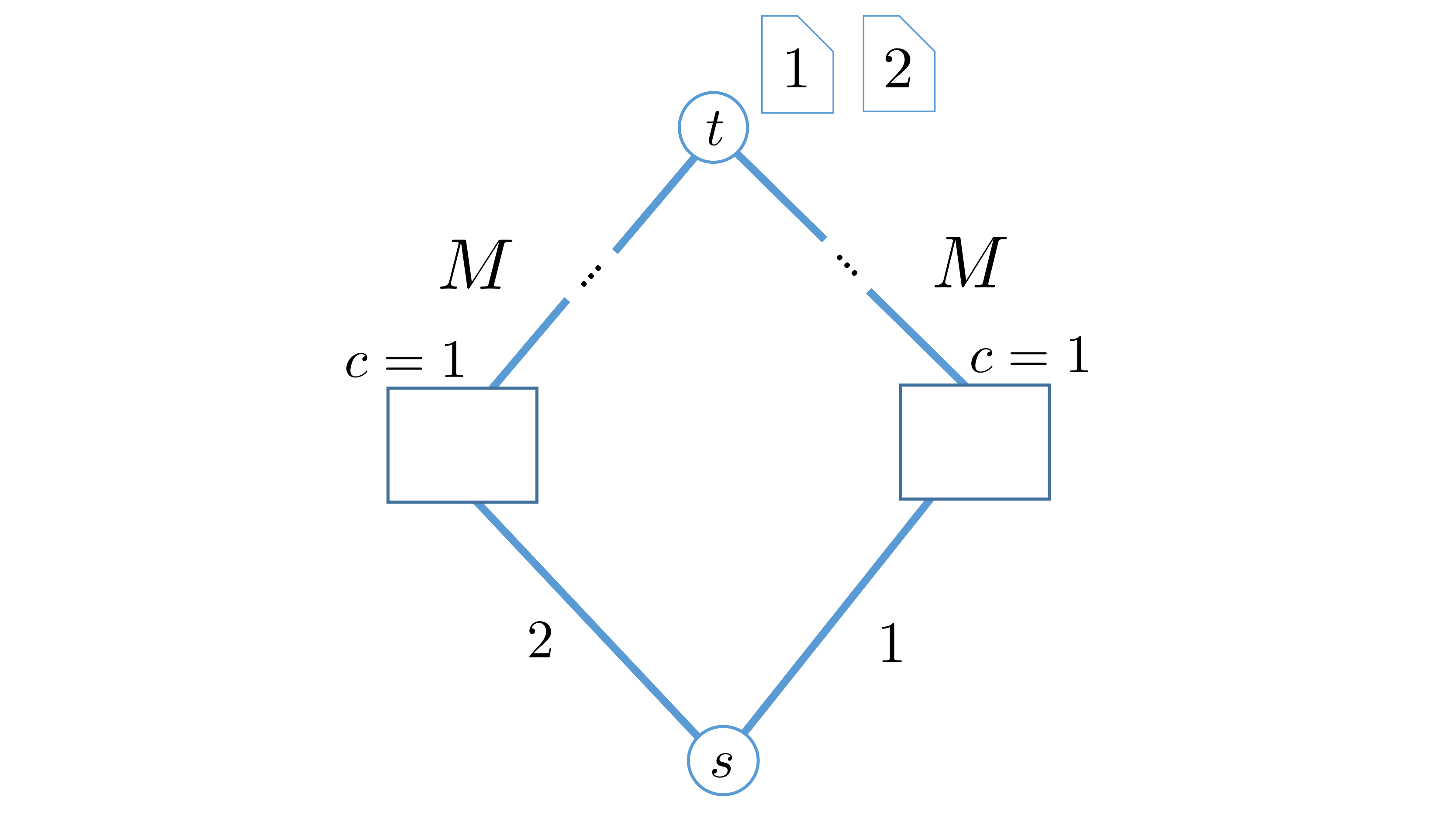}
 \end{center}
 \caption{ \change{A simple example illustrating the benefits of path diversity. A source node $s$ generates requests for items $1$ and $2$,  permanently stored on designated server $t$. Intermediate nodes on the are two alternative paths towards $t$ have capacity 1. Numbers above edges indicate costs. \techrep{ Under RNS, requests for both items are forwarded over the same  path  towards $t$, leading to an $\Omega(M)$ routing cost irrespective of the caching strategy. In contrast the optimal solution uses different paths per item, leading to an $O(1)$ cost.}{}} } \label{fig:simpleexample}
 \end{figure}

\subsection{Offline Source Routing.}
\noindent\textbf{Expected Caching Gain.}
Before presenting a distributed, adaptive joint routing and caching algorithm, we first turn our attention to the offline problem \textsc{MinCost}. As in the solution by \citep{ioannidis2016adaptive} described in Section~\ref{sec:fixedrouting},  we cast this first as a maximization problem. Let $C_0$ be the constant:
\begin{align}C^0_{\SR} =   \textstyle\sum_{(i,s) \in \requests }\lambda_{(i,s)}\sum_{p\in \pathset_{(i,s)}}  \sum_{k=1}^{|p|-1} w_{p_{k+1}p_k}.\end{align}
Then, given a pair of strategies $(r,X)$, we define the \emph{expected caching gain} $F_{\SR}(r,X)$  as follows:
\begin{align}
F_{\SR}(r,X) &= C^0_{\SR}-C_{\SR}(r,X),\label{gains} 
\end{align}
where $C_{\SR}$ is the  aggregate routing cost given by \eqref{routingcost}.
Note that $C^0_{\SR}$ upper bounds the expected routing cost, so that $F_{\SR}(r,X)\geq 0$. We seek to solve the following problem, equivalent to \textsc{MinCost}: \begin{subequations}\label{deterministicsource}
\center{\CGS}\vspace*{-0.5em}
\begin{align}
\text{Maximize:}& \quad F_{\SR}(r,X)\\
\text{subj.~to:}& \quad (r,X)\in \feasibledomain_\SR
\end{align}
\end{subequations}
The selection of the constant $C_\SR^0$ is not arbitrary: this is precisely the value that \change{allows us to approximate  $F_\SR$ via the concave relaxation $L_\SR$ below--c.f.~Eq.~\eqref{myapprox}}.

\noindent\textbf{Approximation Algorithm.} Its equivalence to \textsc{MinCost} implies that \CGS is also NP-hard. 
Nevertheless, we show that there exists a polynomial time approximation algorithm for \CGS: \begin{theorem}\label{thm:offline}
There exists an algorithm that terminates within a number of steps that is polynomial in $|V|$, $|\catalog|$, and $P_\SR$, and produces a strategy $(r',X')\in \feasibledomain_\SR$ such that
\begin{align}F_{\SR}(r',X') \geq (1-1/e)  \textstyle\max_{(r,X)\in \feasibledomain_\SR} F_{\SR}(r,X).\end{align}
\end{theorem}
\techrep{The proof can be found in Appendix~\ref{proofofthmoffline}.}{}
  In Sec.~\ref{sec:evaluation} we show that, in spite of attaining approximation guarantees w.r.t.~$F_\SR$ rather than $C_\SR$, the resulting approximation algorithm has excellent performance in practice in terms of minimizing routing costs. In particular, we can reduce routing costs by a factor as high as $10^3$ compared to fixed routing policies, including \citep{ioannidis2016adaptive}. 
\change{
We  briefly describe the algorithm below, leaving details to \techrep{the appendix.}{\citep{arxivthis}.} Consider the concave function $L_{\SR}:\conv(\feasibledomain_{\SR})\to \reals_+$,   defined as:
\begin{align}
\begin{split}
L_{\SR} (\rho,\Xi) = \textstyle \sum_{(i,s) \in \requests }\lambda_{(i,s)}\sum_{p\in \pathset_{(i,s)}}  \sum_{k=1}^{|p|-1} w_{p_{k+1}p_k} \cdot \\
\textstyle\min\big\{1,1-\rho_{(i,s),p}+\sum_{k'=1}^k \xi_{p_{k'}i}\big\}.
\end{split}
\end{align}
Then, $L_\SR$  closely approximates  $F_{\SR}$ \techrep{(see Lemma~\ref{goodapprox})}{\citep{arxivthis}}:\begin{align}\left(1-{1}/{e}\right)  L_{\SR} (\rho,\Xi) \leq F_{\SR}(\rho,\Xi)\leq L_{\SR} (\rho,\Xi),  \label{myapprox}\end{align}for all $(\rho,\Xi)\in \conv(\feasibledomain_\SR).$
Our  constant-approximation algorithm for \CGS comprises two steps. \emph{First,} obtain
\begin{align}(\rho^*,\Xi^*)\in\textstyle\argmax_{(\rho,\Xi)\in \conv(\feasibledomain_\SR) } L_{\SR} (\rho,\Xi).\label{myrelax}\end{align}
 As $L_{\SR}$ is a concave function and  $\conv(\feasibledomain_\SR)$ is convex, the above maximization is a convex optimization problem. In  fact, it  can be  reduced to a linear program, so it can be solved in  polynomial time \citep{papadimitriou1982combinatorial}. \emph{Second,} round the (possibly fractional) solution $(\rho^*,\Xi^*)\in \conv(\feasibledomain_\SR)$  to an integral solution $(r,X)\in \feasibledomain_\SR$ such that $F_\SR(r,X)\geq F_\SR(\rho^*,\Xi^*)$. This rounding is deterministic and also takes place in polynomial time. }

\change{The rounding technique used in our} proof of Thm.~\ref{thm:offline}  has the following immediate implication\techrep{, whose proof is in Appendix ~\ref{proofofcorrnr}:}{:}
\begin{corollary} \label{cor:rnr}There exists an optimal solution $(r^*,X^*)$ to \CGS (and hence, to \textsc{MinCost}-\SR) in which $r^*$ is an route-to-nearest-replica (RNR)  strategy w.r.t.~$X^*$. 
\end{corollary}
Although,  in light of Theorem~\ref{thm:subopt}, Corollary~\ref{cor:rnr} suggests an advantage of RNR  over RNS strategies, we note it does not give  any intuition on how to construct  an optimal RNR solution.
 
\noindent\textbf{Equivalence of Deterministic and Randomized Strategies.} We can also show the following  result \techrep{(see Appendix~\ref{proofofequiv})}{} regarding \emph{randomized} strategies. For $\mu$ a probability distribution over $\feasibledomain_\SR$, let $\expect_\mu[C_\SR(r,X)]$ be the expected routing cost under $\mu$. Then, the following equivalence theorem holds:
\begin{theorem}\label{cor:equiv}   The deterministic and randomized versions of \textsc{MinCost-\SR} attain the same optimal routing cost, i.e.:
\begin{align}\begin{split}\min_{(r,X)\in \feasibledomain_\SR}C_\SR(r,X)&=
\min_{(\rho,\Xi)\in \conv(\feasibledomain_\SR) }C_\SR(\rho, \Xi)\\
&= \min_{\mu:\supp(\mu)=\feasibledomain_\SR}\expect_\mu[C_\SR(r,X)]  \end{split}\end{align} 
\end{theorem} 
The first equality of the theorem implies that, surprisingly, there is \emph{no inherent advantage in randomization}: although  randomized strategies constitute a superset of deterministic strategies, the optimal attainable  routing cost (or, equivalently, caching gain) is the same  for both classes. The second equality implies that assuming independent caching and routing strategies is as powerful as sampling routing and caching strategies from an arbitrary joint distribution.  Thm.~\ref{cor:equiv} generalizes Thm.~5 of \citep{ioannidis2016adaptive}, which pertains to optimizing caching alone. 
\subsection{Online Source Routing}
The algorithm in Thm.~\ref{thm:offline} is offline and centralized: it assumes full knowledge of the input, including demands and arrival rates, which are rarely a priori available in practice. To that end,
we turn our attention to solving \CGS in the online setting, in the absence of any a priori knowledge of the demand, and seek an algorithm that is both adaptive and distributed.

As described in \ref{sec:offon}, in the online setting, time is partitioned into slots of equal length $T>0$. Caching and routing strategies are randomized as described in Sec.~\ref{sec:model}: at the beginning of a timeslot,  nodes place a random set of contents in their cache, independently of each other; upon  arrival,  a new request is routed over a random path, selected independently of (a) all past routes followed,  and (b) of caching decisions. 
Our next theorem  shows that, in steady state, the expected caching gain of the jointly constructed routing and caching strategies is within a constant approximation of the optimal solution to the offline problem \CGS: 
\begin{theorem}\label{maincor}
There exists a distributed, adaptive algorithm under which the randomized strategies sampled during the $k$-th slot  $(r^{(k)},X^{(k)})\in \feasibledomain_\SR$ satisfy 
\begin{align}\lim_{k\to \infty} \expect[F_\SR(r^{(k)}\!,\!X^{(k)})] \geq (1 - {1}/{e})\max_{(r,X)\in \feasibledomain_\SR} F_\SR(r,X). \end{align}
\end{theorem}
\techrep{The proof can be found in Appendix~\ref{proofofmaincor}.}{} Note that, despite the fact that the algorithm has no prior knowledge of the demands, the guarantee provided is w.r.t.~an optimal solution of the \emph{offline} problem \eqref{deterministicsource}. Moreover, in light of Thm.~\ref{cor:equiv}, our adaptive algorithm is $1-\tfrac{1}{e}$-competitive w.r.t.~optimal offline \emph{randomized} strategies as well. Our algorithm naturally generalizes  \citep{ioannidis2016adaptive}: when the path sets $\pathset_{(i,s)}$ are singletons, and routing is fixed, our algorithm coincides with the cache-only optimization algorithm in \citep{ioannidis2016adaptive}. Interestingly, the algorithm casts routing and caching  \emph{in the same control plane}: the same quantities are communicated through control messages to adapt both the caching and routing strategies.

\begin{algorithm}[t]
\begin{footnotesize}
  \caption{\textsc{Projected Gradient Ascent}}\label{alg:ascent}
    \begin{algorithmic}[1]
   \STATE Execute the following for each $v\in V$ and each $(i,s)\in \requests$:
   \STATE Pick arbitrary state $(\rho^{(0)},\Xi^{(0))}\in \conv(\feasibledomain_\SR)$.
   \FOR{ each timeslot $k\geq 1$}
    \FOR{ each $v\in V$} 
    \STATE Compute the sliding average $\bar{\xi}_v^{(k)}$\techrep{ through \eqref{slidexi}.}{.}
    \STATE Sample a feasible $x^{(k)}_v$ from a distribution with marginals $\bar{\xi}_v^{(k)}$.
    \STATE Place items $x^{(k)}_v$ in cache. \medskip
        \STATE Collect measurements and, at the end of the timeslot, compute estimate \techrep{$z_{v}$}{} of $\partial_{\xi_v} L_\SR(\rho^{k},\Xi^{(k)})$\techrep{ through \eqref{zestimation}.}{.}
    \STATE Adapt \techrep{ through \eqref{adapt}.}{} to new state  $\xi_{v}^{(k+1)}$ in the direction of the gradient with step-size $\gamma_k$, projecting back to $\conv(\feasibledomain_\SR).$ 
  \ENDFOR
   \FOR{ each $(i,s)\in \requests$} 
    \STATE Compute the sliding average $\bar{\rho}_{(i,s)}^{(k)}$\techrep{ through \eqref{sliderho}.}{.}\medskip
    \STATE Whenever a new request arrives, sample $p\in\pathset_{(i,s)}$ from distribution $\bar{\rho}_{(i,s)}^{(k)}$.
        \STATE Collect measurements and, at the end of the timeslot, compute estimate \techrep{$q_{(i,s)}$}{} of $\partial_{\rho_{(i,s)}} L_\SR(\rho^{k},\Xi^{(k)})$\techrep{through \eqref{qestimation}.}{.}
    \STATE Adapt \techrep{through \eqref{adapt}.}{} to new state  $\rho_{(i,s)}^{(k+1)}$ in the direction of the gradient with step-size $\gamma_k$, projecting back to $ \conv(\feasibledomain_\SR)$.
  \ENDFOR
  \ENDFOR

  \end{algorithmic}
\end{footnotesize}
\end{algorithm}

\change{
\noindent\textbf{Algorithm Overview.}
 We give a brief overview of the  distributed, adaptive algorithm that attains the guarantees of Theorem~\ref{maincor} below. The algorithm is summarized in Algorithm~\ref{alg:ascent}. Recall that time is partitioned into slots of equal length $T>0$. Caching and routing strategies are randomized as described in Sec.~\ref{sec:model}: at the beginning of a timeslot,  nodes place a random set of contents in their cache, independently of each other; upon  arrival,  a new request is routed over a random path, selected independently of (a) all past routes followed,  and (b) of caching decisions. 
}

\change{More specifically, nodes in the network maintain the following state information. Each node $v\in G$ maintains  locally a vector
    $\xi_{v}\in[0,1]^{|\catalog|}$, determining its randomized caching strategy. Moreover, for each request $(i,s)\in \requests$, source node $s$ maintains a vector $\rho_{(i,s)} \in[0,1]^{|\pathset_{(i,s)}|}$, determining its randomized routing strategy.  Together, these variables represent the global state of the network, denoted by $(\rho,\Xi) \in \conv(\feasibledomain_\SR)$.  
When the timeslot ends, each node performs the following four tasks:
\begin{packedenumerate}
\item \textbf{Subgradient Estimation}. Each node uses measurements collected during the duration of a timeslot to construct estimates of the gradient of $L_\SR$ w.r.t. its own local state variables. As $L_\SR$ is not everywhere differentiable, an estimate of a \emph{subgradient}  of $L_\SR$  is computed instead. 
\item \textbf{State Adaptation.} Nodes adapt their local caching and routing state variables  $\xi_{v}$, $v\in V$, and $\rho_{(i,s)}$, $(i,s)\in \requests$, pushing them towards a direction that increases $L_\SR$, as determined by the estimated subgradients.
\item \textbf{State Smoothening.}  Nodes compute ``smoothened'' versions $\bar{\xi}_v$, $v\in V$, and $\bar{\rho}_{(i,s)}$, $(i,s)\in \requests$, interpolated between present and past states. This is needed on account of the non-differentiability of $L_{\SR}$. 
\item \textbf{Randomized Caching and Routing.} After smoothening,  each node $v$ reshuffles the contents of its cache using the smoothened caching marginals $\bar{\xi}_v$, producing a random placement (i.e., caching strategy $x_v$) to be used throughout the next slot. Moreover, each node $s\in V$ routes requests  $(i,s)\in \requests$ received during next timeslot over random paths (i.e., routing strategies $r_{(i,s)}$) sampled in an i.i.d.~fashion from the smoothened  marginals $\bar{\rho}_{(i,s)}$.
\end{packedenumerate}
 Together, these steps ensure that, in steady state, the expected caching gain of the jointly constructed routing and caching strategies is within a constant approximation of the optimal solution to the offline problem \CGS. We formally describe  the constituent subgradient estimation, state adaptation,  smoothening, and random sampling steps in  detail in \techrep{Appendices~\ref{sec:distributedsub} to \ref{sec:randomizedrounding}, respectively.}{\citep{arxivthis}.} 
}

\sloppy
\noindent\change{\textbf{Convergence Guarantees.} The proof of the convergence of the algorithm relies on the following key lemma: 
\begin{lemma} \label{convergencelemma}Let $(\bar{\rho}^{(k)},\bar{\Xi}^{(k)})\in \feasibledomain_2$ be the smoothened state variables at the $k$-th slot of Algorithm~\ref{alg:ascent}, and 
$(\rho^*,\Xi^*)\in\argmax_{(\rho,\Xi)\in \conv(\feasibledomain_\SR)} L_\SR(\rho,\Xi).
$
Then, for $\gamma_k$ the step-size used in projected gradient ascent, \begin{align*}\varepsilon_k&\equiv \expect[L_\SR(\rho^*,\Xi^*)-L_\SR(\bar{\rho}^{(k)},\bar{\Xi}^{(k)})] \leq  \frac{D^2 + M^2 \sum_{\ell=\lfloor k/2\rfloor}^{k}\gamma_\ell^2 }{2\sum_{\ell=\lfloor k/2\rfloor}^{k}\gamma_\ell} ,\end{align*}
where $D = \sqrt{2|V|\max_{v\in V}c_v+2|\requests|},$ and $$M=W|V| \Lambda\sqrt{(|V||\catalog|P^2 +|\requests|P) (1+\frac{1}{\Lambda T})}.$$ In particular, $\varepsilon_k= O(1/\sqrt{k})$ for $\gamma = 1/\sqrt{k}$. \end{lemma}
Lemma~\ref{convergencelemma} establishes that Algorithm~\ref{alg:ascent} converges arbitrarily close to an optimizer of $L_\SR$. As, by \eqref{myapprox}, this is a close approximation of $F_\SR$, the limit points of the algorithm are with the $1-1/e$ from the optimal. Crucially, Lemma~\ref{convergencelemma} can be used to determine the rate of convergence of the algorithm, by determining the number of steps required for $\varepsilon_k$ to reach a desired threshold $\delta$. Moreover, through quantity $M$, Lemma~\ref{convergencelemma} establishes a tradeoff w.r.t.~ $T$: increasing $T$ decreases the error in the estimated subgradient, thereby reducing the total number of steps till convergence, but also increases the time taken by each step.
}

\begin{comment}
In particular, Theorem~\ref{maincor} is an immediate consequence of the following lemma:
\begin{lemma}\label{mainlemma}
Let $(r^{(k)},X^{(k)})\in \feasibledomain_\SR$ be the random strategies sampled during the $k$-th slot of Algorithm~\ref{alg:ascent}. Then, 
$$\lim_{k\to \infty} \expect[F_\SR(r^{(k)}\!,\!X^{(k)})] \geq (1 - {1}/{e})\max_{(r,X)\in \feasibledomain_\SR} F_\SR(r,X). $$
\end{lemma}
The proof can be found in  Appendix~\ref{proofofmainlemma}. Before we prove this, 
we  formally describe in Sections~\ref{sec:distributedsub} to \ref{sec:randomizedrounding} the constituent subgradient estimation, state adaptation,  smoothening, and random sampling steps in  detail. An modification of the algorithm, leading to reduced control traffic, at an increase of the corresponding variance of the subgradient estimate, can be found in Section \ref{controlreduction}

\end{comment}.

\begin{comment}Intuitively, our algorithm  solves a problem of the form:
\begin{align}\min_{(\rho,\Xi)\in \conv(\feasibledomain_\SR) } L_\SR(\rho,\Xi)\label{myrelax},\end{align}
where function $L_{\SR}:\conv(\feasibledomain_{\SR})\to \reals_+$ is a \emph{concave approximation} of the expected caching gain $F_\SR$. That is, it is (a) a concave function, and (b) it closely approximates $F_\SR$; in particular, it satisfies:
\begin{align}\left(1-{1}/{e}\right)  L_{\SR} (\rho,\Xi) \leq F_{\SR}(\rho,\Xi)\leq L_{\SR} (\rho,\Xi),\label{myapprox}\end{align}
for all $(\rho,\Xi)\in \conv(\feasibledomain_\SR)$. As a result, in contrast to \eqref{relaxedmin}, \eqref{myrelax} is a convex optimization problem.

Our distributed adaptive algorithm effectively performs a \emph{projected gradient ascent} to solve the convex relaxation \eqref{myrelax} in a distributed, adaptive fashion.  The concavity of $L_{\SR}$ ensures convergence,  while \eqref{myapprox}  ensures that the caching gain attained in steady state is within an $1-\tfrac{1}{e}$ factor from the optimal. 
 
 \end{comment}

\subsection{Hop-by-Hop Routing}

A similar analysis to the one we outlined above applies to hop-by-hop routing, both in the offline and online setting. We state again the main theorems here; proofs can again be found in \techrep{ Appendix~\ref{app:hopbyhop}.}{\citep{arxivthis}.}

\noindent\textbf{Offline Setting.} 
As in the case of source routing,
we define the constant:
$C^0_{\HH} =   \textstyle\sum_{(i,s) \in \requests }\lambda_{(i,s)}\sum_{(u,v)\in G^{(i,s)}} w_{vu} | \pathset_{(i,s)}^u|.$
Using this constant, we define the caching gain maximization problem to be: \begin{subequations}\label{deterministichopbyhop2}
\center{\CGHH}\vspace*{-0.5em}
\begin{align}
\text{Maximize:}& \quad F_{\HH}(r,X)\\
\text{subj.~to:}& \quad (r,X)\in \feasibledomain_\HH
\end{align}
\end{subequations}
where 
$F_{\HH}(r,X) = C^0_{\HH}-\textstyle\sum_{(i,s) \in \requests }\lambda_{(i,s)}  C^{(i,s)}_{\HH}(r,X)$ 
is the expected caching gain.
This is again an NP-hard problem, equivalent to \eqref{deterministichopbyhop}. We can again construct a constant approximation algorithm for \CGHH:
\begin{theorem}\label{thm:offlinehop}
There exists an algorithm that terminates within a number of steps that is polynomial in $|V|$, $|\catalog|$, and $P_\HH$, and produces a strategy $(r',X')\in \feasibledomain_\HH$ such that
\begin{align}
F_{\HH}(r',X') \geq (1-1/e)  \max_{(r,X)\in \feasibledomain_\HH} F_{\HH}
(r,X).\end{align}
\end{theorem}

\noindent\textbf{Online Setting.} Finally, as in the case of source routing, we can provide a distributed, adaptive algorithm for hop-by-hop routing as well.
\begin{theorem}\label{maincor2}
There exists a distributed, adaptive algorithm under which the randomized strategies sampled during the $k$-th slot $(r^{(k)},X^{(k)})\in \feasibledomain_\HH$ satisfy
\begin{align}\lim_{k\to \infty} \expect[F_\HH(r^{(k)},X^{(k)})] \geq \big(1 -{1}/{e}\big)\max_{(r,X)\in \feasibledomain_\SR} F_\HH(r,X). \end{align}
\end{theorem}
We note again that the distributed, adaptive algorithm attains an expected caching gain within a constant approximation from the \emph{offline} optimal.

\section{Evaluation}\label{sec:evaluation}

\begin{figure*}[!t]
\hspace*{4mm}\includegraphics[width=0.98\textwidth]{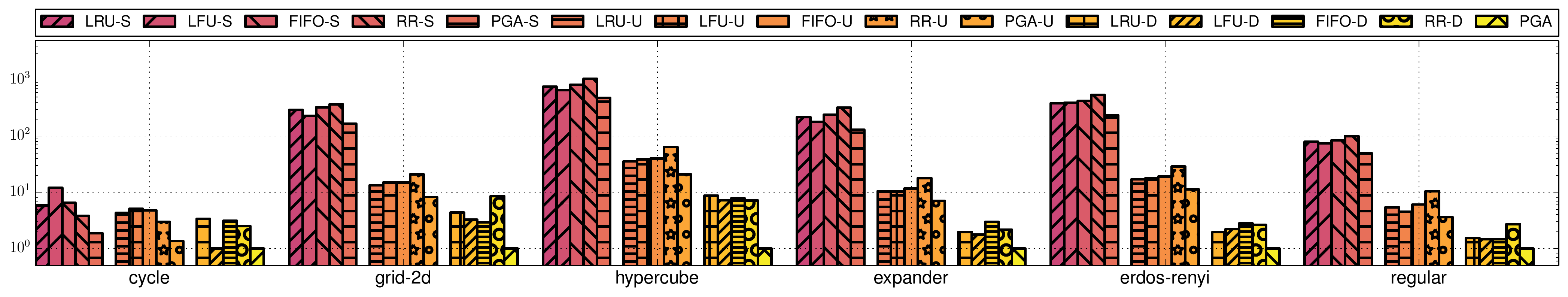}
\put(-330,20){\rotatebox{90}{\tiny $\bar{C}_{\SR}/\bar{C}_\SR^{\texttt{PGA}} $}}
\\
\hspace*{4mm}\includegraphics[width=0.98\textwidth]{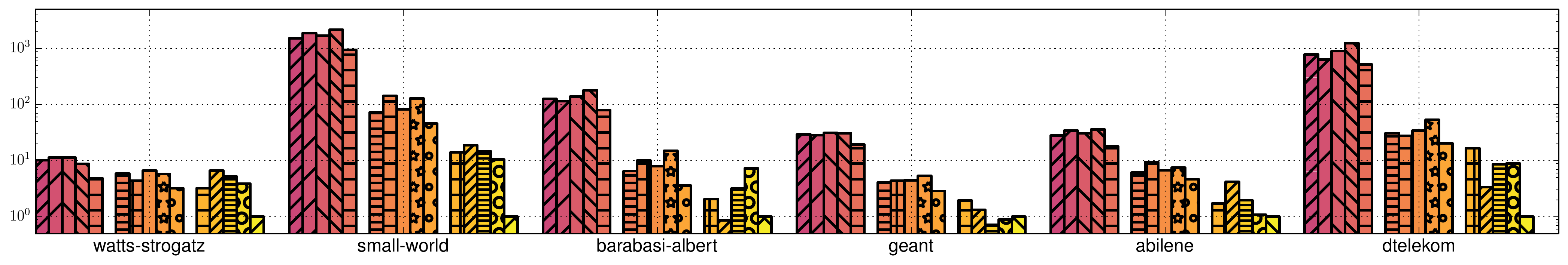}
\put(-330,20){\rotatebox{90}{\tiny $\bar{C}_{\SR}/\bar{C}_\SR^{\texttt{PGA}} $}}
\caption{\change{Ratio of expected routing cost $\bar{C}_\SR$ to routing cost $\bar{C}_\SR^{\texttt{PGA}}$ under our $\texttt{PGA}$ policy},  for different topologies and  strategies. For each topology, each of the three groups of bars corresponds to a routing strategy, namely, RNS/shortest path routing (\texttt{-S}), uniform routing (\texttt{-U}), and dynamic routing (\texttt{-D}). The algorithm presented in \citep{ioannidis2016adaptive} is \texttt{PGA-S}, while our algorithm (\texttt{PGA}), with ratio 1.0, is shown last for reference purposes; \change{values of of  $\bar{C}_\SR^{\texttt{PGA}}$ are given in Table~\ref{networks}.}}\label{relativeperformance}
\end{figure*}

We simulate \techrep{Algorithm \ref{alg:ascent}}{our distributed, adaptive algorithm for \CGS} over a broad variety of both synthetic and real networks. We compare its performance to traditional caching policies,  combined with both static and dynamic multi-path routing. 
\begin{table}[!t]
\begin{center}
\begin{footnotesize}
\begin{tabular}{l|p{1em}p{1em}p{1em}p{1em}p{1em}p{1em}cc}
\hline
Graph & $|V|$ & $|E|$ & $|\catalog|$ & $|\requests|$ &  $|Q|$ & $c_v$ & $|\pathset_{(i,s)}|$ & $\bar{C}_\SR^{\texttt{PGA}}$ \\
\hline
\texttt{cycle}  & 30 & 60 & 10 & 100 & 10  & 2  & 2  &  20.17   \\
 \texttt{grid-2d} &100 & 360  & 300 & 1K & 20 & 3& 30 & 0.228\\
\texttt{hypercube}&128 & 896  & 300 & 1K & 20 & 3& 30 &  0.028\\
 \texttt{expander} &100 & 716  & 300 & 1K & 20 & 3& 30 & 0.112\\
\hline
 \texttt{erdos-renyi} &100 & 1042  & 300 & 1K & 20 & 3& 30 &  0.047\\
\texttt{regular} &100 & 300  & 300 & 1K & 20 & 3& 30 & 0.762 \\
 \texttt{watts-strogatz} &100 & 400  & 300 & 1K & 20 & 3& 2 & 35.08 \\
 \texttt{small-world} &100 & 491  & 300 & 1K & 20 & 3& 30 & 0.029\\
 \texttt{barabasi-albert} &100 & 768  & 300 & 1K & 20 & 3& 30 & 0.187 \\
\hline
 \texttt{geant} &22 & 66& 10&100 & 10 & 2 & 10 & 1.28 \\
 \texttt{abilene} & 9 & 26 & 10 &90 & 9 & 2 & 10 & 0.911 \\
 \texttt{dtelekom} & 68 & 546 & 300 & 1K & 20 & 3 & 30 & 0.025 \\
\hline
\end{tabular}
\end{footnotesize}
\end{center}
\caption{Graph Topologies and Experiment Parameters.}\label{networks}
\end{table}
\noindent\textbf{Experiment Setup.} We consider the topologies  in Table~\ref{networks}. 
For each graph $G(V,E)$  in Table~\ref{networks}, we generate a catalog of size $|\catalog|$, and assign to each node $v\in V$ a cache of capacity $c_v$. For every item $i\in \catalog$, we designate a node  selected u.a.r.~from $V$ as a designated server for this item; the item is  stored outside the designate server's cache. We assign a weight to each edge in $E$ selected u.a.r.~from the interval $[1,100]$. We also select a random set of $Q$ nodes as the possible request sources, and generate a set of requests $\requests \subseteq \catalog\times V$ by sampling exactly $|\requests|$ from the set $\catalog \times Q$, uniformly at random. For each such request $(i,s)\in \requests$, we select the request rate $\lambda_{(i,s)}$ according to a Zipf distribution with parameter $1.2$; these are normalized so that   average request rate over all $|Q|$ sources is 1 request per time unit.
For each request $(i,s)\in \requests$, we generate $|\pathset_{(i,s)}|$ paths from the source $s\in V$ to the designated server of item $i\in \catalog$. In all cases, this path set includes the shortest path to the designated server. We consider only paths with stretch at most 4.0; that is, the maximum cost of a path in $P_{(i,s)}$ is at most 4 times the cost of the shortest path to the designated source.
The values of $|\catalog|$, $|\requests|$ $|Q|$, $c_v$, and $\pathset_{(i,s)}$ for each experiment are given in Table~\ref{networks}. 
\noindent\textbf{Online Caching and Routing Algorithms.} We compare the performance of our joint caching and routing projected gradient ascent algorithm (\texttt{PGA}) to several competitors. In terms of caching, we consider four traditional eviction policies for comparison: Least-Recently-Used (\texttt{LRU}), Least-Frequently-Used (\texttt{LFU}), First-In-First-Out (\texttt{FIFO}), and Random Replacement (\texttt{RR}). We combine these policies with path-replication~\citep{cohen2002replication,jacobson2009networking}:  once a request for an item reaches a cache that stores the item, every cache in the reverse path on the way to the query source stores the item, evicting stale items using one of the above eviction policies. We combine the above caching policies with three different routing policies. In route-to-nearest-server (\texttt{-S}), only the shortest path to the nearest designated server is used to route the message. In uniform routing (\texttt{-U}), the source $s$ routes each request $(i,s)$ on a path selected uniformly at random among all paths in $\pathset_{(i,s)}$. We combine each of these (static) routing strategies with each of the above caching strategies use.  For instance, \texttt{LRU-U} indicates \texttt{LRU} evictions combined with uniform routing. Note that \texttt{PGA-S}, i.e., our algorithm restricted to RNS routing, is exactly the single-path routing algorithm proposed in~\citep{ioannidis2016adaptive}.
To move beyond static routing policies for \texttt{LRU}, \texttt{LFU}, \texttt{FIFO}, and \texttt{RR}, we also combine the above traditional caching strategies with an adaptive routing strategy,  akin to our algorithm, with  estimates of the expected routing cost at each path used to adapt routing strategies.
 During a slot, each source node $s$ maintains an average  of the routing cost incurred when routing a request over each path. At the end of the slot, the source decreases the probability $\rho_{(i,s),p}$ that it will follow the path $p$ by an amount proportional to the average, and projects the new strategy to the simplex. \change{For fixed caching strategies, this dynamic routing scheme converges to a route-to-nearest-replica (RNS) routing, which we expect by Cor.~\ref{cor:rnr} to have good performance.} We denote this routing scheme with the extension \texttt{-D}. \change{Note that all algorithms we simulate are online.}

\noindent\textbf{Experiments and Measurements}. Each experiment consists of a simulation of the caching and routing policy, over a specific topology, for a total of 5000 time units. To leverage PASTA, we collect measurements during the duration of the execution at exponentially distributed intervals with mean 1.0 time unit. At each measurement epoch, we extract the current cache contents in the network and construct $X\in \{0,1\}^{|V|\times |\catalog|}$. Similarly, we extract the current routing strategies $\rho_{(i,s)}$ for all requests $(i,s)\in \requests$, and construct the global routing strategy $\rho\in [0,1]^{P_\SR }$.  Then, we evaluate the \emph{expected routing cost} $C_\SR(\rho,X)$.  We report the average $\bar{C}_\SR$ of these values across measurements collected after a warmup phase, during 1000 and 5000 time units of the simulation; \change{that is, if $t_i$ are the measurement times, then $\bar{C}_{\SR} = \frac{1}{t_\texttt{tot}-t_{\texttt{w}}} \sum_{t_i:\in [t_{\texttt{w}},t_{\texttt{tot}}]}C_{\SR}(\rho(t_i),X(t_i)) .$}

\noindent\textbf{Performance w.r.t~Routing Costs.} The relative performance of the different strategies to our algorithm is shown in Figure~\ref{relativeperformance}. With the exception of \texttt{cycle} and \texttt{watts-strogatz}, where paths are scarce, we see several common trends across topologies. First, simply moving from RNS routing to uniform, multi-path routing,  reduces the routing cost by a factor of 10. Even without optimizing routing or caching, simply increasing path options increases the available caching capacity. For all caching policies, optimizing routing through the dynamic routing policy (denoted by \texttt{-D}), reduces routing costs by another factor of 10. Finally, jointly optimizing routing and caching leads to a reduction by an additional factor between 2 and 10 times. In several cases, \texttt{PGA} outperforms RNS routing combined with either traditional policies or \citep{ioannidis2016adaptive} by 3 orders of magnitude.

\techrep{
\begin{table*}[!t]
\begin{tiny}
\begin{tabular}{l | p{1em}p{1em}p{1em}p{1em}p{2.5em} |p{1em}p{1em}p{1em}p{1em}p{2.0em}  |  p{1em}p{1em}p{1em}p{1em}p{2.0em}   }
\hline
Graph & \texttt{LRU-S} & \texttt{LFU-S} & \texttt{FIFO-S} & \texttt{RR-S}  & \texttt{PGA-S} & \texttt{LRU-U} & \texttt{LFU-U} & \texttt{FIFO-U} & \texttt{RR-U}  & \texttt{PGA-U} & \texttt{LRU} & \texttt{LFU} & \texttt{FIFO} & \texttt{RR} &  \texttt{PGA}\\
\hline
\texttt{cycle} &0.47 & 0.47 & 4.38 & 0.47 & 865.29 & 0.47 & 0.47 & 3.62 & 0.47  & 436.14 & 6.62 & 24.08 & 12.91 & 0.47  & 148.20\\
\texttt{grid-2d} &0.08 & 3.77 & 0.08 & 0.53 & 657.84 & 0.08 & 0.08  & 0.08 & 0.08 & 0.08 & 0.08 & 0.08 & 0.08 & 0.08 &  0.08\\
\texttt{hypercube} &0.21 & 2.27 & 1.62 & 19.48 & 924.75 & 0.21 & 0.21 & 0.21  & 0.21 & 0.21 & 0.21 & 0.21 & 0.21 & 0.21  & 0.21\\
\texttt{expander} &0.38 & 1.31 & 0.38 & 70.02  & 794.27 & 0.38 & 0.38 & 0.38  & 0.38 & 0.38 & 0.38 & 0.38 & 0.38 & 0.38 & 0.38\\
\hline
\texttt{erdos-renyi} &3.08 & 0.25 & 0.93 & 5.40  & 870.84 & 0.25 & 0.25 & 0.25  & 0.25 & 0.25 & 0.25 & 0.25 & 0.25 & 0.25  & 0.25\\
\texttt{regular} &1.50 & 7.53 & 1.64 & 13.66  & 1183.97 & 0.05 & 0.05 & 0.05 & 0.05  & 8.52 & 0.05 & 0.05 & 2.54 & 1.40  & 11.49\\
\texttt{watts-strogatz}\!\!\! &11.88 & 8.15 & 19.47 & 19.47  & 158.39 & 7.80 & 2.88 & 9.01 & 8.22  & 54.90 & 19.22 & 7.01 & 14.00 & 17.79  & 37.05\\
\texttt{small-world} &0.30 & 1.02 & 0.30 & 20.07  & 955.48 & 0.30 & 0.30 & 0.30 & 0.30  & 0.30 & 0.30 & 0.30 & 0.30 & 0.30 & 0.30 \\
\texttt{barabasi-albert}\!\!\! &1.28 & 1.28 & 1.28 & 2.12  & 1126.24 & 1.28 & 1.28 & 1.28 & 1.28  & 6.86 & 1.28 & 1.28 & 1.28 & 1.28  & 7.58\\
\hline
\texttt{geant} &0.09 & 0.09 & 20.75 & 4.13 & 1312.96 & 1.85 & 0.09 & 0.09 & 0.09  & 12.71 & 0.09 & 13.34 & 0.09 & 7.88  & 14.41\\
\texttt{abilene} &3.44 & 3.44 & 3.44 & 3.44  & 802.66 & 3.44 & 3.44 & 3.44 & 3.44  & 23.08 & 5.75 & 6.13 & 6.04 & 7.09  & 14.36\\
\texttt{dtelekom} &0.30 & 0.30 & 0.30 & 0.42  & 927.24 & 0.30 & 0.30 & 0.30 & 0.30 & 0.30 & 0.30 & 0.30 & 0.30 & 0.30 & 0.30 \\
\hline
\end{tabular}
\end{tiny}
\caption{Convergence times, in simulation time units.}\label{convergencetable}
\end{table*}
}{\begin{table}
\begin{footnotesize}
\begin{tabular}{l | cc | cc  | cc  }
\hline
Graph & \texttt{LRU-S} & \texttt{PGA-S} & \texttt{LRU-U} & \texttt{PGA-U} & \texttt{LRU} & \texttt{PGA}\\
\hline
\texttt{cycle} &0.47 & 865.29 & 0.47 & 436.14 & 6.62 & 148.20\\
\texttt{grid-2d} &0.08 & 657.84 & 0.08 & 0.08 & 0.08 & 0.08\\
\texttt{hypercube} &0.21 & 924.75 & 0.21 & 0.21 & 0.21 & 0.21\\
\texttt{expander} &0.38 & 794.27 & 0.38 & 0.38 & 0.38 & 0.38\\
\hline
\texttt{erdos-renyi} &3.08 & 870.84 & 0.25 & 0.25 & 0.25 & 0.25\\
\texttt{regular} &1.50 & 1183.97 & 0.05 & 8.52 & 0.05 & 11.49\\
\texttt{watts-strogatz} &11.88 & 158.39 & 7.80 & 54.90 & 19.22 & 37.05\\
\texttt{small-world} &0.30 & 955.48 & 0.30 & 0.30 & 0.30 & 0.30\\
\texttt{barabasi-albert} &1.28 & 1126.24 & 1.28 & 6.86 & 1.28 & 7.58\\
\hline
\texttt{geant} &0.09 & 1312.96 & 1.85 & 12.71 & 0.09 & 14.41\\
\texttt{abilene} &3.44 & 802.66 & 3.44 & 23.08 & 5.75 & 14.36\\
\texttt{dtelekom} &0.30 & 927.24 & 0.30 & 0.30 & 0.30 & 0.30\\
\hline
\end{tabular}
\end{footnotesize}
\caption{Convergence times, in simulation time units, for \texttt{LRU} and \texttt{PGA} caching strategies with different routing variants. Total simulation time is 5K time units. In almost all cases, convergence to steady state occurs much faster than our warm-up period (1K time units).}
\label{convegencetable2}
\end{table}
}

\noindent\textbf{Convergence.} In Table~\techrep{\ref{convergencetable}}{\ref{convegencetable2}}, we show the convergence time for \techrep{all algorithms we simulated.}{the different variants of LRU and PGA-convergence times for other algorithms can be found in our techreport \citep{arxivthis}.} We define the convergence time to be the time at which the time-average caching gain reaches $95\%$ of the expected caching gain attained at steady state. LRU converges faster than PGA, though it converges to a sub-optimal stationary distribution. Interestingly, both \texttt{-U} and adaptive routing reduce convergence times for PGA,  in some cases (like \texttt{grid-2d} and \texttt{dtelekom}) to the order of magnitude of LRU: this is because path diversification reduces contention: it assigns contents to non-overlapping caches, which are populated quickly with distinct contents.
 \section{Conclusions}\label{sec:conclusions}
We have constructed joint caching and routing schemes with optimality guarantees for arbitrary network topologies. Identifying schemes that lead to improved approximation guarantees, especially on the routing cost directly rather than on the caching gain, is an important open question. Equally important is to incorporate queuing and congestion. In particular, accounting for queueing delays and identifying delay-minimizing strategies is open even under fixed routing. Such an analysis can also potentially be used to understand how different caching and routing schemes affect both delay optimality and  throughput optimality.
 \bibliographystyle{plainnat}
\bibliography{references} 

\begin{thebibliography}{58}
\providecommand{\natexlab}[1]{#1}
\providecommand{\url}[1]{\texttt{#1}}
\expandafter\ifx\csname urlstyle\endcsname\relax
  \providecommand{\doi}[1]{doi: #1}\else
  \providecommand{\doi}{doi: \begingroup \urlstyle{rm}\Url}\fi

\bibitem[Abedini and Shakkottai(2014)]{abedini2014content}
Navid Abedini and Srinivas Shakkottai.
\newblock Content caching and scheduling in wireless networks with elastic and
  inelastic traffic.
\newblock \emph{IEEE/ACM Transactions on Networking}, 22\penalty0 (3):\penalty0
  864--874, 2014.

\bibitem[Achlioptas et~al.(2000)Achlioptas, Chrobak, and
  Noga]{achlioptas2000competitive}
Dimitris Achlioptas, Marek Chrobak, and John Noga.
\newblock Competitive analysis of randomized paging algorithms.
\newblock \emph{Theoretical Computer Science}, 234\penalty0 (1):\penalty0
  203--218, 2000.

\bibitem[Ageev and Sviridenko(2004)]{ageev2004pipage}
Alexander~A Ageev and Maxim~I Sviridenko.
\newblock Pipage rounding: A new method of constructing algorithms with proven
  performance guarantee.
\newblock \emph{Journal of Combinatorial Optimization}, 8\penalty0
  (3):\penalty0 307--328, 2004.

\bibitem[Applegate et~al.(2010)Applegate, Archer, Gopalakrishnan, Lee, and
  Ramakrishnan]{applegate2010optimal}
David Applegate, Aaron Archer, Vijay Gopalakrishnan, Seungjoon Lee, and
  Kadangode~K Ramakrishnan.
\newblock Optimal content placement for a large-scale {VoD} system.
\newblock In \emph{CoNext}, 2010.

\bibitem[Baev et~al.(2008)Baev, Rajaraman, and Swamy]{baev2008approximation}
Ivan Baev, Rajmohan Rajaraman, and Chaitanya Swamy.
\newblock Approximation algorithms for data placement problems.
\newblock \emph{SIAM Journal on Computing}, 38\penalty0 (4):\penalty0
  1411--1429, 2008.

\bibitem[Barab{\'a}si and Albert(1999)]{barabasi1999emergence}
Albert-L{\'a}szl{\'o} Barab{\'a}si and R{\'e}ka Albert.
\newblock Emergence of scaling in random networks.
\newblock \emph{Science}, 286\penalty0 (5439):\penalty0 509--512, 1999.

\bibitem[Bartal et~al.(1995)Bartal, Fiat, and Rabani]{bartal1995competitive}
Yair Bartal, Amos Fiat, and Yuval Rabani.
\newblock Competitive algorithms for distributed data management.
\newblock \emph{Journal of Computer and System Sciences}, 51\penalty0
  (3):\penalty0 341--358, 1995.

\bibitem[Batista et~al.(2007)Batista, Da~Fonseca, and Miyazawa]{batista2007set}
Daniel~M Batista, Nelson~LS Da~Fonseca, and Flavio~K Miyazawa.
\newblock A set of schedulers for grid networks.
\newblock In \emph{Proceedings of the 2007 ACM symposium on Applied computing},
  pages 209--213. ACM, 2007.

\bibitem[Berger et~al.(2014)Berger, Gland, Singla, and Ciucu]{berger2014exact}
Daniel~S Berger, Philipp Gland, Sahil Singla, and Florin Ciucu.
\newblock Exact analysis of {TTL} cache networks.
\newblock \emph{IFIP Performance}, 2014.

\bibitem[B{\l}aszczyszyn and Giovanidis(2015)]{wireless-cache}
Bart{\l}omiej B{\l}aszczyszyn and Anastasios Giovanidis.
\newblock Optimal geographic caching in cellular networks.
\newblock In \emph{Proc. of ICC}, London, UK, 2015.

\bibitem[Borst et~al.(2010)Borst, Gupta, and Walid]{borst2010distributed}
Sem Borst, Varun Gupta, and Anwar Walid.
\newblock Distributed caching algorithms for content distribution networks.
\newblock In \emph{INFOCOM}, 2010.

\bibitem[Calinescu et~al.(2007)Calinescu, Chekuri, P{\'a}l, and
  Vondr{\'a}k]{calinescu2007maximizing}
Gruia Calinescu, Chandra Chekuri, Martin P{\'a}l, and Jan Vondr{\'a}k.
\newblock Maximizing a submodular set function subject to a matroid constraint.
\newblock In \emph{Integer programming and combinatorial optimization}, pages
  182--196. Springer, 2007.

\bibitem[Carofiglio et~al.(2016)Carofiglio, Mekinda, and
  Muscariello]{carofiglio2016joint}
Giovanna Carofiglio, L{\'e}once Mekinda, and Luca Muscariello.
\newblock Joint forwarding and caching with latency awareness in
  information-centric networking.
\newblock \emph{Computer Networks}, 110:\penalty0 133--153, 2016.

\bibitem[Che et~al.(2002)Che, Tung, and Wang]{che2002hierarchical}
Hao Che, Ye~Tung, and Zhijun Wang.
\newblock Hierarchical web caching systems: Modeling, design and experimental
  results.
\newblock \emph{Selected Areas in Communications}, 20\penalty0 (7):\penalty0
  1305--1314, 2002.

\bibitem[Chiocchetti et~al.(2012)Chiocchetti, Rossi, Rossini, Carofiglio, and
  Perino]{chiocchetti2012exploit}
Raffaele Chiocchetti, Dario Rossi, Giuseppe Rossini, Giovanna Carofiglio, and
  Diego Perino.
\newblock Exploit the known or explore the unknown?: {Hamlet-like} doubts in
  icn.
\newblock In \emph{Proceedings of the second edition of the ICN workshop on
  Information-centric networking}, pages 7--12. ACM, 2012.

\bibitem[Cohen and Shenker(2002)]{cohen2002replication}
Edith Cohen and Scott Shenker.
\newblock Replication strategies in unstructured peer-to-peer networks.
\newblock In \emph{SIGCOMM}, 2002.

\bibitem[Cormen et~al.(2009)Cormen, Leiserson, Rivest, and
  Stein]{cormen2009book}
T~Cormen, C~Leiserson, R~Rivest, and C~Stein.
\newblock \emph{Introduction to Algorithms}.
\newblock MIT Press, 2009.

\bibitem[Dan and Towsley(1990)]{dan1990approximate}
Asit Dan and Don Towsley.
\newblock An approximate analysis of the lru and fifo buffer replacement
  schemes.
\newblock In \emph{SIGMETRICS}, volume~18. ACM, 1990.

\bibitem[Dehghan et~al.(2014)Dehghan, Seetharam, Jiang, He, Salonidis, Kurose,
  Towsley, and Sitaraman]{dehghan2014complexity}
Mostafa Dehghan, Anand Seetharam, Bo~Jiang, Ting He, Theodoros Salonidis, Jim
  Kurose, Don Towsley, and Ramesh Sitaraman.
\newblock On the complexity of optimal routing and content caching in
  heterogeneous networks.
\newblock In \emph{INFOCOM}, 2014.

\bibitem[Dehghan et~al.(2015)Dehghan, Massoulie, Towsley, Menasche, and
  Tay]{dehghan2016utility}
Mostafa Dehghan, Laurent Massoulie, Don Towsley, Daniel Menasche, and YC~Tay.
\newblock A utility optimization approach to network cache design.
\newblock In \emph{INFOCOM}, 2015.

\bibitem[Fayazbakhsh et~al.(2013)Fayazbakhsh, Lin, Tootoonchian, Ghodsi,
  Koponen, Maggs, Ng, Sekar, and Shenker]{fayazbakhsh2013less}
Seyed~Kaveh Fayazbakhsh, Yin Lin, Amin Tootoonchian, Ali Ghodsi, Teemu Koponen,
  Bruce Maggs, KC~Ng, Vyas Sekar, and Scott Shenker.
\newblock Less pain, most of the gain: Incrementally deployable icn.
\newblock In \emph{ACM SIGCOMM Computer Communication Review}, volume~43, pages
  147--158. ACM, 2013.

\bibitem[Flajolet et~al.(1992)Flajolet, Gardy, and
  Thimonier]{flajolet1992birthday}
Philippe Flajolet, Daniele Gardy, and Lo{\"y}s Thimonier.
\newblock Birthday paradox, coupon collectors, caching algorithms and
  self-organizing search.
\newblock \emph{Discrete Applied Mathematics}, 39\penalty0 (3):\penalty0
  207--229, 1992.

\bibitem[Fleischer et~al.(2006)Fleischer, Goemans, Mirrokni, and
  Sviridenko]{fleischer2006tight}
Lisa Fleischer, Michel~X Goemans, Vahab~S Mirrokni, and Maxim Sviridenko.
\newblock Tight approximation algorithms for maximum general assignment
  problems.
\newblock In \emph{SODA}, 2006.

\bibitem[Fofack et~al.(2012)Fofack, Nain, Neglia, and
  Towsley]{fofack2012analysis}
N~Choungmo Fofack, Philippe Nain, Giovanni Neglia, and Don Towsley.
\newblock Analysis of {TTL}-based cache networks.
\newblock In \emph{VALUETOOLS}, 2012.

\bibitem[Fricker et~al.(2012)Fricker, Robert, and
  Roberts]{fricker2012versatile}
Christine Fricker, Philippe Robert, and James Roberts.
\newblock A versatile and accurate approximation for {LRU} cache performance.
\newblock In \emph{ITC}, 2012.

\bibitem[Gabber and Galil(1981)]{gabber1981explicit}
Ofer Gabber and Zvi Galil.
\newblock Explicit constructions of linear-sized superconcentrators.
\newblock \emph{Journal of Computer and System Sciences}, 22\penalty0
  (3):\penalty0 407--420, 1981.

\bibitem[Gelenbe(1973)]{gelenbe1973unified}
Erol Gelenbe.
\newblock A unified approach to the evaluation of a class of replacement
  algorithms.
\newblock \emph{IEEE Transactions on Computers}, 100\penalty0 (6):\penalty0
  611--618, 1973.

\bibitem[Goemans and Williamson(1994)]{goemans1994new}
Michel~X Goemans and David~P Williamson.
\newblock New 3/4-approximation algorithms for the maximum satisfiability
  problem.
\newblock \emph{SIAM Journal on Discrete Mathematics}, 7\penalty0 (4):\penalty0
  656--666, 1994.

\bibitem[Guenter et~al.(2011)Guenter, Jain, and Williams]{guenter2011managing}
Brian Guenter, Navendu Jain, and Charles Williams.
\newblock Managing cost, performance, and reliability tradeoffs for
  energy-aware server provisioning.
\newblock In \emph{INFOCOM, 2011 Proceedings IEEE}, pages 1332--1340. IEEE,
  2011.

\bibitem[Ioannidis and Yeh(2016)]{ioannidis2016adaptive}
Stratis Ioannidis and Edmund Yeh.
\newblock Adaptive caching networks with optimality guarantees.
\newblock In \emph{ACM SIGMETRICS}, 2016.

\bibitem[Jacobson et~al.(2009)Jacobson, Smetters, Thornton, Plass, Briggs, and
  Braynard]{jacobson2009networking}
Van Jacobson, Diana~K Smetters, James~D Thornton, Michael~F Plass, Nicholas~H
  Briggs, and Rebecca~L Braynard.
\newblock Networking named content.
\newblock In \emph{CoNEXT}, 2009.

\bibitem[Jiang et~al.(2012)Jiang, Lan, Ha, Chen, and Chiang]{jiang2012joint}
Joe~Wenjie Jiang, Tian Lan, Sangtae Ha, Minghua Chen, and Mung Chiang.
\newblock Joint vm placement and routing for data center traffic engineering.
\newblock In \emph{INFOCOM, 2012 Proceedings IEEE}, pages 2876--2880. IEEE,
  2012.

\bibitem[King(1971)]{thomas1971analysis}
WF~King.
\newblock Analysis of paging algorithms.
\newblock Technical report, Thomas J. Watson IBM Research Center, 1971.

\bibitem[Kleinberg(2000)]{kleinberg2000small}
Jon Kleinberg.
\newblock The small-world phenomenon: An algorithmic perspective.
\newblock In \emph{STOC}, 2000.

\bibitem[Kurose and Ross(2007)]{kurose2007computer}
James~F Kurose and Keith~W Ross.
\newblock \emph{Computer networking: a top-down approach}.
\newblock Addison Wesley, 2007.

\bibitem[Laoutaris et~al.(2004)Laoutaris, Syntila, and
  Stavrakakis]{laoutaris2004meta}
Nikolaos Laoutaris, Sofia Syntila, and Ioannis Stavrakakis.
\newblock Meta algorithms for hierarchical web caches.
\newblock In \emph{ICPCC}, 2004.

\bibitem[Laoutaris et~al.(2006)Laoutaris, Che, and
  Stavrakakis]{laoutaris2006lcd}
Nikolaos Laoutaris, Hao Che, and Ioannis Stavrakakis.
\newblock The lcd interconnection of lru caches and its analysis.
\newblock \emph{Performance Evaluation}, 63\penalty0 (7):\penalty0 609--634,
  2006.

\bibitem[Li et~al.(2011)Li, Tordsson, and Elmroth]{li2011virtual}
Wubin Li, Johan Tordsson, and Erik Elmroth.
\newblock Virtual machine placement for predictable and time-constrained peak
  loads.
\newblock In \emph{International Workshop on Grid Economics and Business
  Models}, pages 120--134. Springer, 2011.

\bibitem[Lv et~al.(2002)Lv, Cao, Cohen, Li, and Shenker]{lv2002search}
Qin Lv, Pei Cao, Edith Cohen, Kai Li, and Scott Shenker.
\newblock Search and replication in unstructured peer-to-peer networks.
\newblock In \emph{ICS}, 2002.

\bibitem[Martina et~al.(2014)Martina, Garetto, and
  Leonardi]{martina2014unified}
Valentina Martina, Michele Garetto, and Emilio Leonardi.
\newblock A unified approach to the performance analysis of caching systems.
\newblock In \emph{INFOCOM}, 2014.

\bibitem[Michelot(1986)]{michelot1986finite}
Christian Michelot.
\newblock A finite algorithm for finding the projection of a point onto the
  canonical simplex of $\reals^n$.
\newblock \emph{Journal of Optimization Theory and Applications}, 50\penalty0
  (1):\penalty0 195--200, 1986.

\bibitem[Naveen et~al.(2015)Naveen, Massouli{\'e}, Baccelli, Carneiro~Viana,
  and Towsley]{naveen2015interaction}
KP~Naveen, Laurent Massouli{\'e}, Emmanuel Baccelli, Aline Carneiro~Viana, and
  Don Towsley.
\newblock On the interaction between content caching and request assignment in
  cellular cache networks.
\newblock In \emph{ATC}, 2015.

\bibitem[Nemirovski(2005)]{nemirovski2005efficient}
Arkadi Nemirovski.
\newblock \emph{Efficient methods in convex programming}.
\newblock 2005.

\bibitem[Papadimitriou and Steiglitz(1982)]{papadimitriou1982combinatorial}
Christos~H Papadimitriou and Kenneth Steiglitz.
\newblock \emph{Combinatorial optimization: algorithms and complexity}.
\newblock Courier Corporation, 1982.

\bibitem[Poularakis et~al.(2013)Poularakis, Iosifidis, and
  Tassiulas]{poularakis2013approximation}
Konstantinos Poularakis, George Iosifidis, and Leandros Tassiulas.
\newblock Approximation caching and routing algorithms for massive mobile data
  delivery.
\newblock In \emph{GLOBECOM}, 2013.

\bibitem[Psaras et~al.(2012)Psaras, Chai, and Pavlou]{psaras2012probabilistic}
Ioannis Psaras, Wei~Koong Chai, and George Pavlou.
\newblock Probabilistic in-network caching for information-centric networks.
\newblock In \emph{Proceedings of the second edition of the ICN workshop on
  Information-centric networking}, pages 55--60. ACM, 2012.

\bibitem[Rosensweig et~al.(2010)Rosensweig, Kurose, and
  Towsley]{rosensweig2010approximate}
Elisha~J Rosensweig, Jim Kurose, and Don Towsley.
\newblock Approximate models for general cache networks.
\newblock In \emph{INFOCOM, 2010 Proceedings IEEE}, pages 1--9. IEEE, 2010.

\bibitem[Rosensweig et~al.(2013)Rosensweig, Menasche, and
  Kurose]{rosensweig2013steady}
Elisha~J Rosensweig, Daniel~S Menasche, and Jim Kurose.
\newblock On the steady-state of cache networks.
\newblock In \emph{INFOCOM}, 2013.

\bibitem[Rossi and Rossini(2011)]{rossi2011caching}
Dario Rossi and Giuseppe Rossini.
\newblock Caching performance of content centric networks under multi-path
  routing (and more).
\newblock Technical report, Telecom ParisTech, 2011.

\bibitem[Rossini and Rossi(2014)]{rossini2014coupling}
Giuseppe Rossini and Dario Rossi.
\newblock Coupling caching and forwarding: Benefits, analysis, and
  implementation.
\newblock In \emph{Proceedings of the 1st international conference on
  Information-centric networking}, pages 127--136. ACM, 2014.

\bibitem[Shanmugam et~al.(2013)Shanmugam, Golrezaei, Dimakis, Molisch, and
  Caire]{shanmugam2013femtocaching}
Karthikeyan Shanmugam, Negin Golrezaei, Alexandros~G Dimakis, Andreas~F
  Molisch, and Giuseppe Caire.
\newblock Femtocaching: Wireless content delivery through distributed caching
  helpers.
\newblock \emph{Transactions on Information Theory}, 59\penalty0 (12):\penalty0
  8402--8413, 2013.

\bibitem[Van~den Bossche et~al.(2010)Van~den Bossche, Vanmechelen, and
  Broeckhove]{van2010cost}
Ruben Van~den Bossche, Kurt Vanmechelen, and Jan Broeckhove.
\newblock Cost-optimal scheduling in hybrid iaas clouds for deadline
  constrained workloads.
\newblock In \emph{Cloud Computing (CLOUD), 2010 IEEE 3rd International
  Conference on}, pages 228--235. IEEE, 2010.

\bibitem[Vondr{\'a}k(2008)]{vondrak2008optimal}
Jan Vondr{\'a}k.
\newblock Optimal approximation for the submodular welfare problem in the value
  oracle model.
\newblock In \emph{STOC}, 2008.

\bibitem[Wang et~al.(2013)Wang, Li, Tyson, Uhlig, and Xie]{wang2013optimal}
Yonggong Wang, Zhenyu Li, Gareth Tyson, Steve Uhlig, and Gaogang Xie.
\newblock Optimal cache allocation for content-centric networking.
\newblock In \emph{2013 21st IEEE International Conference on Network Protocols
  (ICNP)}, pages 1--10. IEEE, 2013.

\bibitem[Watts and Strogatz(1998)]{watts1998collective}
Duncan~J Watts and Steven~H Strogatz.
\newblock Collective dynamics of `small-world' networks.
\newblock \emph{Nature}, 393\penalty0 (6684):\penalty0 440--442, 1998.

\bibitem[Xie et~al.(2012)Xie, Shi, and Wang]{xie2012tecc}
Haiyong Xie, Guangyu Shi, and Pengwei Wang.
\newblock Tecc: Towards collaborative in-network caching guided by traffic
  engineering.
\newblock In \emph{INFOCOM, 2012 Proceedings IEEE}, pages 2546--2550. IEEE,
  2012.

\bibitem[Yeh et~al.(2014)Yeh, Ho, Cui, Burd, Liu, and Leong]{yeh2014vip}
Edmund Yeh, Tracey Ho, Ying Cui, Michael Burd, Ran Liu, and Derek Leong.
\newblock {VIP}: A framework for joint dynamic forwarding and caching in named
  data networks.
\newblock In \emph{ICN}, 2014.

\bibitem[Zhou et~al.(2004)Zhou, Chen, and Li]{zhou2004second}
Yuanyuan Zhou, Zhifeng Chen, and Kai Li.
\newblock Second-level buffer cache management.
\newblock \emph{Parallel and Distributed Systems}, 15\penalty0 (6):\penalty0
  505--519, 2004.

\end{thebibliography}

%\newpage
 \appendix

\section{Cost Minimization Under Randomized Strategies}\label{app:randomisrelaxed}
We prove here that, under the randomized routing and caching strategies with the independence assumptions described in Section~\ref{sec:model}, problem \eqref{deterministicsourcecost} reduces to:
\begin{align} \min_{(\rho,\Xi)\in \conv(\feasibledomain_\SR)}C_\SR(\rho,\Xi)\label{relaxedmincompact}\end{align}
where 
$\conv(\feasibledomain_\SR)$ is the convex hull of $\feasibledomain_\SR$.
The set  $\conv(\feasibledomain_\SR)$ consists of pairs  $(\rho,\Xi)$ where  
 \begin{align*}
\rho&=[\rho_{(i,s),p}]_{(i,s)\in\requests ,p\in \pathset_{(i,s)} }\in [0,1]^P,\quad\text{and} \\
\Xi&=[\xi_{vi}]_{v\in V,i\in \catalog}\in [0,1]^{|V|\times |\catalog|},
\end{align*}  satisfy \eqref{xcapacity} and \eqref{pselect}.

Consider a randomized routing strategy $r$ and a randomized caching strategy $X$, such that $(r,X)\in \feasibledomain_\SR$. Let  $\rho=\expect[r]$ and $\Xi = \expect[X]$. Then $(r,X)\in \feasibledomain_\SR$ readily implies that $(\rho,\Xi)\in \conv(\feasibledomain_\SR)$. Indeed, as $(r,X)$ satisfy $\eqref{xintegrality}$ and $\eqref{pintegrality}$, there expectiations must have coordinates within $[0,1]$. On the other hand, the constraints   \eqref{xcapacity} and \eqref{pselect} are linear, so if they are satisfied by $(r,X)$ they must be satisfied by their expectations. Hence, any randomized strategy pair  $(r,X)\in \feasibledomain_\SR$ necessarily has expectation $(\rho,\Xi)$ that is in the feasibility 
domain of problem \eqref{relaxedmincompact}; moreover, by \eqref{relaxfun}, its expected routing cost would be exactly given by $C_\SR(\rho,\Xi)$.

To complete the proof, we need to show that any feasible solution $(\rho,\Xi)\in \conv(\feasibledomain_\SR) $ of \eqref{relaxedmincompact}, we can constuct a \textsc{MinCost} feasible pair of randomized strategies $(r,X)\in \feasibledomain_\SR$ whose expectations are precisely $(\rho,\Xi)$; then, by \eqref{relaxfun}, it must be that $\expect[C_\SR(r,X)]=C_\SR(\rho,X)$. Note that this construction is trivial for routing strategies: given $(\rho,\Xi)\in \conv(\feasibledomain_\SR)$, we can construct a randomized strategy $r$ by setting $r_{(i,s)}$ for each $(i,s)\in \requests$  to be an independent categorical variable over $\pathset_{(i,s)}$ with $\prob[r_{(i,s),p}=1]=\rho_{(i,s),p}$, for $p\in \pathset_{(i,s)}$. It is less obvious how to do so for caching strategies; nevertheless,  the technique by \citep{ioannidis2016adaptive,wireless-cache}  presented in Appendix~\ref{sec:randomizedrounding} below (see also Figure~\ref{fig:allocation}) achieves precisely the desired property: given a feasible $\Xi$, it produces a feasible randomized placement $X$, independent across nodes that (a) satisfies capacity constraints, and (b) has marginals given by $\Xi$. \qed

\section{Proof of Theorem~{\protect{\lowercase{\ref{thm:subopt}}}}}\label{app:proofofthmsubopt}
Consider simple diamond network shown in Figure \ref{fig:simpleexample}. 
A source node $s$ generates requests for items $1$ and $2$ (i.e., $\requests=\{(1,s),(2,s)\}$), that are permanently stored on designated server $t$, requesting each with equal rate $\lambda_{(1,s)}=\lambda_{(2,s)}=1\text{sec}^{-1}$. The path sets $\pathset_{(i,s)}$, $i=1,2$, are identical, and consist of the two alternative paths towards $t$, each passing through an intermediate node with cache capacity 1 (i.e., able to store only one item). The two paths have routing costs $M+1$ and $M+2$, respectively.  Under the route-to-nearest server strategy $r'$, requests for both items are forwarded over the path of length $M+1$  towards $t$; fixing routes this way leads to a cost $M+1$ for at least one of the items. This happens irrespectively of which item is cached in the intermediate node; as a result the expected routing cost is $\Theta(M)$. On the other hand, if routing and caching decisions are jointly optimized, requests for the two items can be forwarded to different paths, allowing both items to be cached in the nearby caches, and reducing the cost for both requests to at most 2. 
\qed 

\section{Offline Source Routing}
\subsection{Proof of Theorem~{\protect{\lowercase{\ref{thm:offline}}}}}\label{proofofthmoffline}
Following~\citep{ageev2004pipage}, the technique for producing an approximation algorithm to solve \CGS is to: (a) relax the combinatorial joint routing and caching problem to a convex optimization problem, (b) solve this convex relaxation, and (c) round the (possibly fractional) solution to obtain an integral solution to the original problem. 

To that end, 
consider the concave function $L_{\SR}:\conv(\feasibledomain_{\SR})\to \reals_+$,   defined as:
\begin{align}
\begin{split}
L_{\SR} (\rho,\Xi) = \textstyle \sum_{(i,s) \in \requests }\lambda_{(i,s)}\sum_{p\in \pathset_{(i,s)}}  \sum_{k=1}^{|p|-1} w_{p_{k+1}p_k} \cdot \\
\textstyle\min\big\{1,1-\rho_{(i,s),p}+\sum_{k'=1}^k \xi_{p_{k'}i}\big\}.
\end{split}
\end{align}
Then, $L_\SR$  closely approximates  $F_{\SR}$:
\begin{lemma}\label{goodapprox}
$\left(1-{1}/{e}\right)  L_{\SR} (\rho,\Xi) \leq F_{\SR}(\rho,\Xi)\leq L_{\SR} (\rho,\Xi)$ for all $(\rho,\Xi)\in \conv(\feasibledomain_\SR)$.
\end{lemma}
\begin{proof}
This follows from the Goemans-Williamson inequality \citep{goemans1994new,ioannidis2016adaptive}: for any sequence of $y_i\in [0,1]$, $i\in\{1,\ldots,n\}$,
$$\textstyle(1-{1}/{e})\min\{1,\sum_i^n y_i \} \leq 1- \prod_{i=1}^n(1-y_i)\leq \min\{1,\sum_{i=1}^ny_i\}.$$
The lemma therefore follows by applying the inequality to every term in the summation making up $F_\SR$, to all variables  $\xi_{vi}$, $v\in V$, $i\in \catalog$, and $1-\rho_{(i,s),p}$, $(i,s)\in \requests$, $p\in \pathset_{(i,s)}$.  
\end{proof} 

\begin{algorithm}[t]
\begin{small}
  \caption{\textsc{Offline Algorithm}}\label{alg:sroffline}
  {\fontsize{9}{9}\selectfont
  \begin{algorithmic}[1]
   \STATE Find $(\rho^*,\Xi^*)\in\argmax_{(\rho,\Xi)\in \conv(\feasibledomain_\SR)}L_\SR(\rho,\Xi)$
   \STATE Fix $\rho^*$, and round $\Xi^*$  to obtain integral, feasible $X'$ s.t.~$F_\SR(\rho^*,X')\geq F_\SR(\rho^*,\Xi^*)$
   \STATE Fix $X'$, and round $\rho^*$ to obtain integral, feasible $r'$ s.t.~.~$F_\SR(r',X')\geq F_\SR(\rho^*,X')$

   \RETURN $(r',X')$
  \end{algorithmic}
}
\end{small}
\end{algorithm}

Constructing a constant-approximation algorithm for \CGS amounts to the following steps. \emph{First,} obtain
\begin{align}(\rho^*,\Xi^*)\in\argmax_{(\rho,\Xi)\in \conv(\feasibledomain_\SR) } L_{\SR} (\rho,\Xi).\label{convsolve} \end{align}
 As $L_{\SR}$ is a concave function and  $\conv(\feasibledomain_\SR)$ is convex, the above maximization is a convex optimization problem. In  fact, it  can be  reduced to a linear program, so it can be solved in  polynomial time \citep{papadimitriou1982combinatorial}: for completeness we outline  this reduction in Appendix~\ref{app:linear}. \emph{Second,} round the (possibly fractional) solution $(\rho^*,\Xi^*)\in \conv(\feasibledomain_\SR)$  to an integral solution $(r,X)\in \feasibledomain_\SR$ such that $F_\SR(r,X)\geq F_\SR(\rho^*,\Xi^*)$. This rounding is deterministic and takes place in polynomial time. 
 
The above steps are summarized in Algorithm~\ref{alg:sroffline}.
The following two lemmas hold.  First, a feasible fractional solution can be converted--in polynomial time--to a feasible solution in which only $\rho$ is fractional, while increasing $F_\SR$:
\begin{lemma}[\citep{shanmugam2013femtocaching}]\label{roundcaching}
 Given any $(\rho,\Xi)\in \conv(\feasibledomain_\SR)$, an integral $X$ such that $(\rho,X)\in \conv(\feasibledomain_\SR)$ and $F_\SR(\rho,X)\geq F_\SR(\rho,\Xi)$ can be constructed  in $O(|V|^2|\catalog|P)$ time.  \end{lemma}
\begin{proof}This is proved in \citep{shanmugam2013femtocaching} for fixed routing strategies; for completeness, we repeat the proof here. Given a fractional solution $(\rho,\Xi)\in \feasibledomain_\SR$, there must exist a $v\in V$  that contains two fractional values $\xi_{vi},\xi_{vi'}$; this is because, as $c_v\in \naturals$ for all $v\in V$, the capacity constraints in $\feasibledomain_\SR$ imply that fractional elements of a row in matrix $\Xi$ must come in pairs. Observe that, restricted to these two variables, function $F_\SR$ is an affine function of  $\xi_{vi},\xi_{vi'}$. As such, it is maximized at the extrema of the polytope in $\reals^2$ implied by the capacity and $[0,1]$ constraints involving variables  $\xi_{vi},\xi_{vi'}$. As a result, there is a way to transfer equal mass from one of the two variables to the other so that (a) one of them becomes integral (either 0 or 1), (b) the resulting $\Xi'$ remains feasible,  and (c) $F_\SR$ does not decrease.\footnote{This property is called $\epsilon$-convexity by Ageev and Sviridenko \citep{ageev2004pipage}.} Performing this transfer of mass reduces the number of fractional variables in $\Xi$ by one, while maintaining feasibility and, crucially, either increasing $F_\SR$ or keeping it constant. This rounding can be repeated so long as $\Xi$ remains fractional: this eliminates all fractional variables in at most $O(|V|\catalog)$ steps. Each step requires at most four evaluations of $F_{\SR}$, which can be done in $O(|V| P)$ time. Note that the pair of fractional variables selected each time is arbitrary: the order of elimination (i.e., the order with which pairs of fractional variables are rounded) leads to a different rounding, but all such roundings are (a) feasible and, (b) either increase $F_\SR$ or keep it constant.
\end{proof}
The routing strategy 
   $\rho$ can also be rounded in polynomial time, while keeping the caching strategy $X$ fixed:

\begin{lemma}\label{roundrouting}
\sloppy Given any $(\rho,\Xi)\in \conv(\feasibledomain_\SR)$, an integral $r$ s.t.~$(r,\Xi)\in\conv(\feasibledomain_\SR)$ and $F_\SR(r,\Xi)\geq F_\SR(\rho,\Xi)$ can be constructed in $O(|V|P)$ time. Moreover, if $\Xi$ is integral, then the resulting $r$ is a route-to-nearest-replica (RNR) strategy.  
\end{lemma}
\begin{proof}
 Given $(\rho,\Xi)\in \conv(\feasibledomain_\SR)$, notice that, for fixed $\Xi$, $F_{\SR}$ is an affine function of the routing strategy $\rho$. All coefficients involving variables $\rho_{(i,s),p}$, $p\in \pathset_{(i,s)}$, are non-negative, and the set of constraints on $\rho$ is separable across requests $(i,s)\in\requests$. Hence, given $\Xi$, maximizing $F_\SR$ w.r.t. $\rho$ can be done by selecting the path $p^*\in\pathset_{(i,s)}$ with the highest coefficient of $F_\SR$, for every $(i,s)\in \pathset$; this is precisely the lowest cost path, i.e., $p^*_{(i,s)}\in \pathset_{(i,s)}$ is such that
\begin{align}p^*_{(i,s)}=\argmin_{p\in \pathset_{(i,s)}} \sum_{k=1}^{|p|-1} w_{p_{k+1}p_k}\prod_{k'=1}^k (1-\xi_{p_{k'}i}).\label{selectpstar}\end{align}
Hence, given $\Xi$, setting $\rho_{(i,s),p^*}=1$, and   $\rho_{(i,s),p}=0$ for all remaining paths $p\in\pathset_{(i.s.)}$ s.t. $p\neq p^*$ can only increase $F_\SR$. Each $p^*$
 can be computed in $O(|\pathset_{(i,s)}||V|)$ time and there is most $O(\requests)$ such paths. This results in an integral, feasible strategy $r$,  and the resulting $F_\SR$ either increases or stays constant, i.e., $(r,\Xi)\in \conv(\feasibledomain_\SR)$ and  $F_\SR(r,\Xi)\geq F_\SR(\rho,\Xi)$. 
Finally, if $\Xi=X$ for some integral $X$, then the selection of each strategy $p^*$ through \eqref{selectpstar} yields precisely a route-to-nearest-replica routing for $(i,s)$. Note that, in contrast to rounding the caching strategy in the proof of Lemma~\ref{roundcaching}, the order with which routing strategies are rounded does affect the final integral strategy. \end{proof}

\fussy
 Putting everything together, having $(\rho^*,\Xi^*)$ a solution to \eqref{convsolve}, we can construct an integral solution $(r',X')$ by:
\begin{packedenumerate}
\item fixing the routing strategy to $\rho^*$, and rounding $\Xi^*$ to get an integral $X'$ as in Lemma~\ref{roundcaching}, and then
\item fixing the caching strategy to $X'$, and rounding $\rho^*$ to get an integral $r'$ as in Lemma~\ref{roundrouting}.
\end{packedenumerate}

To conclude the proof of Theorem~\ref{thm:offline}, note that the complexity statement is a consequence of Lemmas~\ref{roundcaching} and~\ref{roundrouting}, and the fact that \eqref{convsolve} reduces to a linear program.
 By construction, the output of the algorithm $(r',X')$ is such that:
  $$F_{\SR}(r',X')\geq  F_\SR(\rho^*,\Xi^*). $$ 
Let $(r^*,X^*) \in  \argmax_{(r,X)\in \feasibledomain_\SR} F_{\SR}(r,X)$ be an optimal solution to \CGS. Then,
by Lemma \ref{goodapprox} and the optimality of $(\rho^*,X^*)$ in $\conv(\feasibledomain_\SR)$:
\begin{align*}F_\SR(r^*,X^*)\leq L_\SR(r^*,X^*) \leq L_\SR(\rho^*,\Xi^*)\leq \frac{e}{e-1} F_\SR(\rho^*,\Xi^*).\end{align*}
Together, these imply that the constructed $(r',X')$ is such that	
$$F_{\SR}(r',X') \geq (1-1/e) F_\SR(r^*,X^*),$$
and the theorem follows.
\qed

\subsection{Reducing \eqref{convsolve} to A Linear Program}\label{app:linear}
We show here how to convert problem \eqref{convsolve} into a linear program. This can be done by introducing appropriate auxiliary variables. In particular, we introduce an auxiliary variable 
$$ t_{(i,s),p,k}, (i,s)\in\requests,p\in\pathset_{(i,s)}, k=1,\ldots,|p|-1 $$ 
for every term in the sum defining function $L_\SR$.  Then, the problem \eqref{convsolve} is equivalent to:
\begin{align*}
\text{Maximize: }&   \sum_{(i,s) \in \requests }\lambda_{(i,s)}\sum_{p\in \pathset_{(i,s)}}  \sum_{k=1}^{|p|-1} w_{p_{k+1}p_k} \cdot t_{(i,s),p,k} \\
\displaybreak[0]
\text{subj.~to: }& 
t_{(i,s),p,k}\leq 1-\rho_{(i,s),p}+\sum_{k'=1}^k \xi_{p_{k'}i}, \text{\quad and} \\
&t_{(i,s),p,k}\leq 1  \text{\quad for all } (i,s)\in\requests,p\in\pathset_{(i,s)},\\
&\qquad\qquad\qquad\qquad k=1,\ldots,|p|-1\\
\displaybreak[0]
& \textstyle\sum_{p\in \pathset_{(i,s)}} \rho_{(i,s),p} =1,\text{\quad for all }(i,s)\in\requests, \\
& \textstyle\sum_{i\in\catalog} \xi_{vi} \leq \capacity_v,\text{\quad for all }v\in V, \\
&\xi_{vi}\in [0,1], \text{ for all }v\in V, i\in\catalog, \text{\quad and} \\
&\rho_{(i,s),p}\in [0,1], \text{\quad for all }p\in\pathset_{(i,s)},(i,s)\in \requests.  
\end{align*}
There are a total of $O(P|V|)$ auxiliary variables, so this is linear program with $O(P|V|)+O(|V||\catalog|)$ variables and $O(P|V|) +O(|\requests|)+O(|V|)+O(|V||\catalog|)+O(P)$ constraints, so it can be solved in time that is in polynomial $P$, $|V|$, and $|\catalog|$ \citep{papadimitriou1982combinatorial}.

\subsection{Proof of Corollary \ref{cor:rnr}}\label{proofofcorrnr}
Let $(r^*,X^*)$ be an  optimal solution to \CGS in which $r^*$ is not a RNR strategy. Then, by Lemma~\ref{roundrouting}, we can construct an $r'$ that is an RNR strategy w.r.t.~$X$ such that (a) $F_\SR(r',X^*)\geq F_\SR(r^*,X^*)$ and (b) $(r',X^*)\in \feasibledomain_\SR$. As $(r^*,X^*)$ is optimal, so is $(r',X^*)$. \qed

\subsection{Proof of Theorem \ref{cor:equiv}}\label{proofofequiv}
Clearly, $$\min_{(r,X)\in \feasibledomain_\SR}C_\SR(r,X) \geq \min_{(\rho,\Xi)\in \conv(\feasibledomain_\SR) }C_\SR(\rho, \Xi)$$ 
as $\feasibledomain_\SR \subset \conv(\feasibledomain_\SR).$ Let $$(\rho^*,\Xi^*) \in \argmin_{(\rho,\Xi)\in \conv(\feasibledomain_\SR) } C_\SR(\rho, \Xi)=  \argmax_{(\rho,\Xi)\in \conv(\feasibledomain_\SR) } F_\SR(\rho, \Xi). $$
Then, Lemmas~\ref{roundcaching} and \ref{roundrouting} imply that we can construct an integral $(r'',X'')\in \feasibledomain_\SR $ s.t.~
\begin{align} F_\SR(r'',X'') \geq   F_\SR(\rho^*,\Xi^*),\label{roundedbetter}\end{align}
which implies that
\begin{align*}
 \min_{(r,X)\in \feasibledomain_\SR}C_\SR(r,X)& \leq C_\SR(r'',X'') \\ &\stackrel{\eqref{roundedbetter}}{\leq}   C_\SR(\rho^*,\Xi^*)\\
&=\min_{(\rho,\Xi)\in \conv(\feasibledomain_\SR) } C_\SR(\rho, \Xi),\end{align*}
and the first equality follows.

To show the second equality, note that for $$\mu^*\in \argmin_{\mu:\supp(\mu)=\feasibledomain_\SR} \expect_\mu[C_\SR(r,X)],$$ and $$(r^*,X^*)=\argmin _{(r,X)\in \feasibledomain_\SR}C_\SR(r,X),$$ we have that \begin{align*}\expect_{\mu^*}[C_\SR(r,X)] &= \min_{\mu:\supp(\mu)=\feasibledomain_\SR} \expect_{\mu}[C_\SR(r,X)]\\ &\leq \min_{(r,X)\in \feasibledomain_\SR}C_\SR(r,X) =C_\SR(r^*,X^*), \end{align*}
as deterministic strategies are a subset of randomized strategies. On the other hand, 
\begin{align*} \expect_{\mu^*}[C_\SR(r,X)]
&=\sum_{(r,X)\in \feasibledomain_\SR} \mu\left((r,X)\right)C_\SR(r,X)\\ 
& \geq C_\SR(r^*,X^*) \sum_{(r,X)\in \feasibledomain_\SR} \mu\left((r,X)\right)=C_\SR(r^*,X^*)\end{align*}
and the second equality also follows.
\qed

\section{ONLINE SOURCE ROUTING}

Before proving Theorem~\ref{maincor}, we  formally describe in Sections~\ref{sec:distributedsub} to \ref{sec:randomizedrounding} the constituent subgradient estimation, state adaptation,  smoothening, and random sampling steps in  detail. A modification of   
 Algorithm~\ref{alg:ascent},
 leading to reduced control traffic, at an increase of the corresponding variance of the subgradient estimate, can be found in Appendix \ref{controlreduction}.

\begin{figure*}
\includegraphics[width=0.32\textwidth]{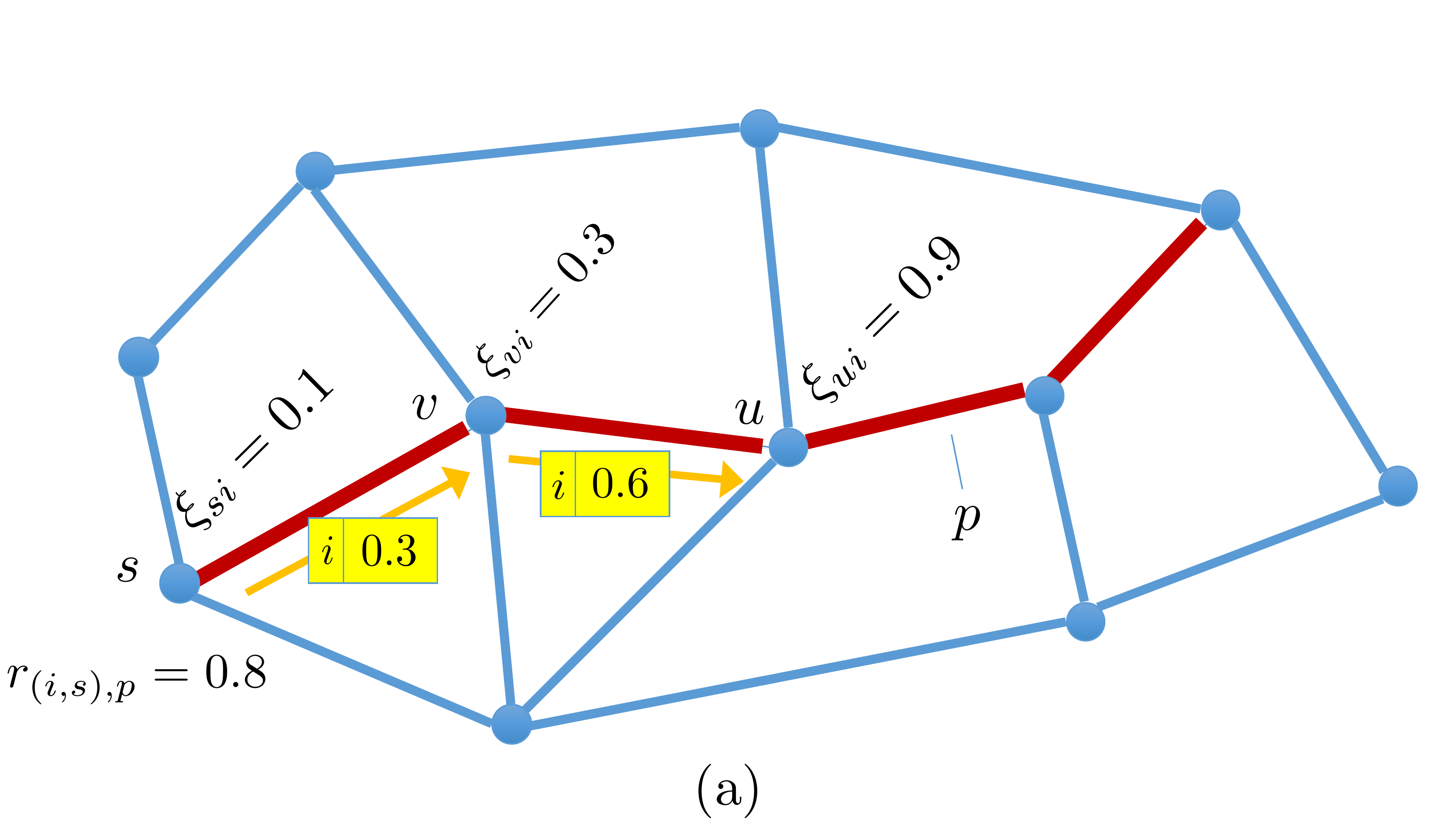}
\includegraphics[width=0.32\textwidth]{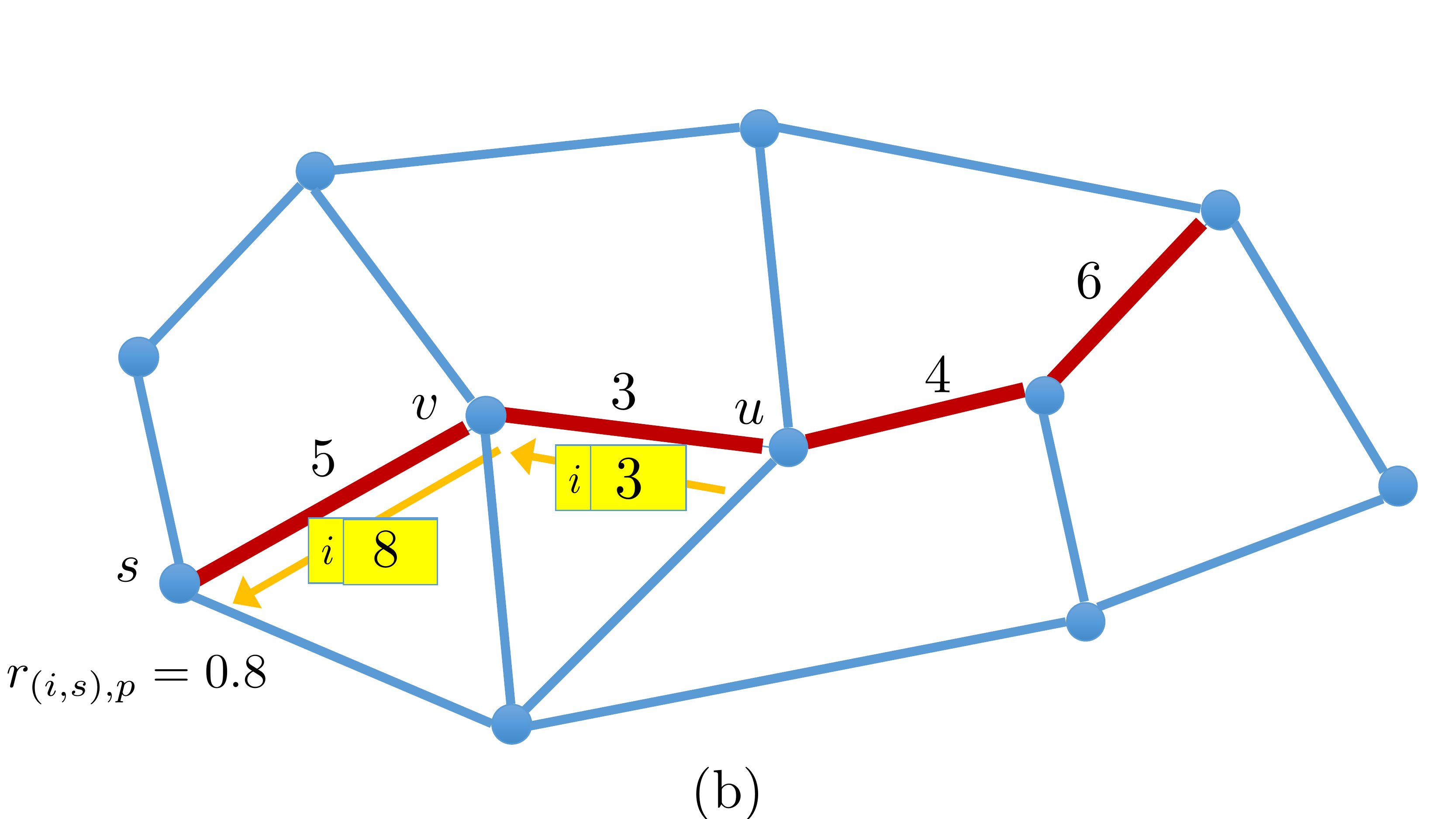}
\includegraphics[width=0.32\textwidth]{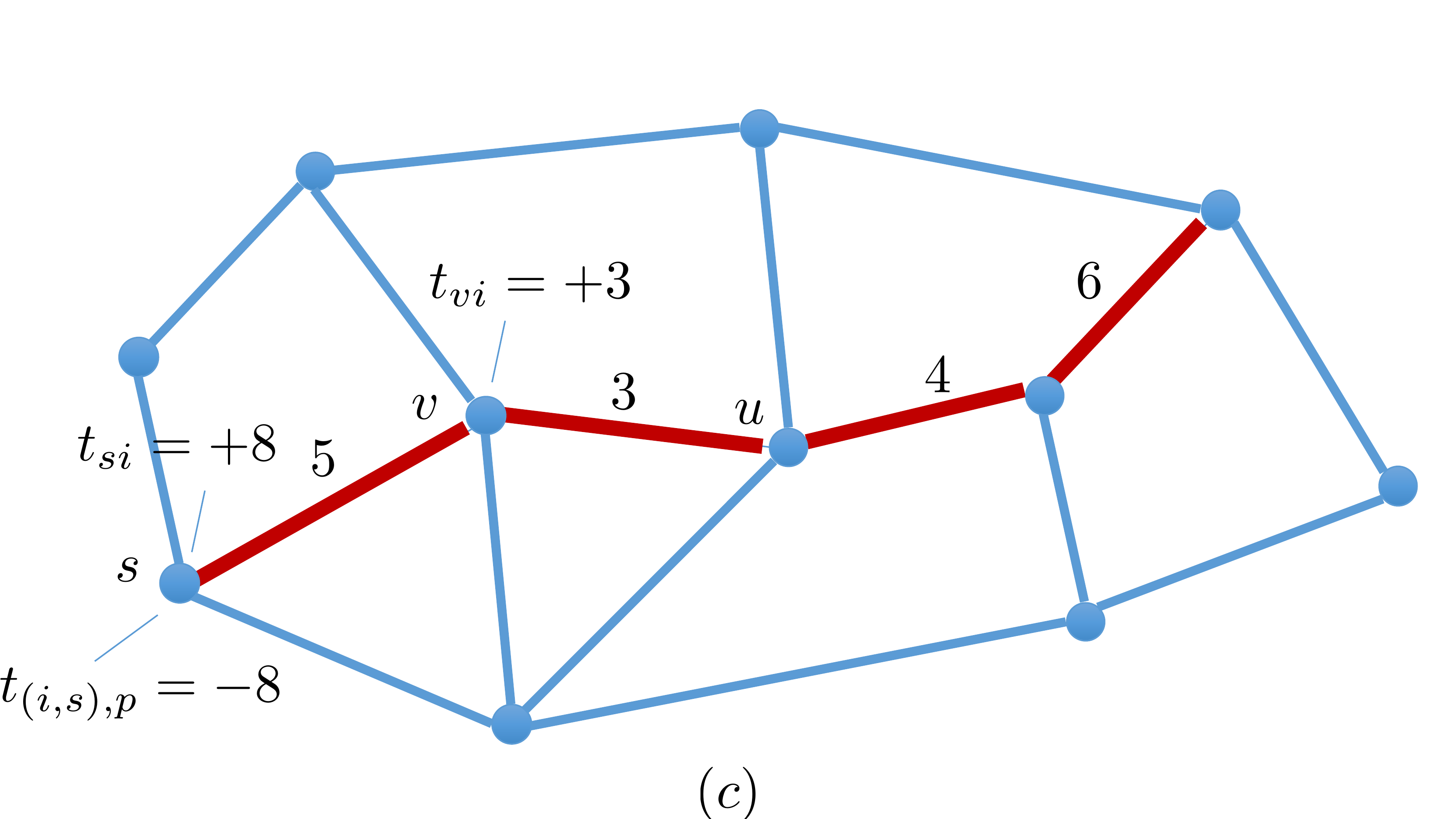}\\
\vspace*{-1em}
\caption{Example of control message trajectory over a single path. When node $s$ generates a request $(i,s)\in \requests$, it also generates a control message for every path $p\in \pathset_{(i,s)}$.  Path $p$ is indicated by red thick edges. In subfigure (a), the control message  counter is initialized to $(1-r_{(i,s),p})+\xi_{si} = 1-0.8+0.1 = 0.3 $ by $s$. It is is forwarded upstream on $p$ to node $v$, that adds its own caching state variable w.r.t.~item $i$, namely $\xi_{ui}=0.3$, to the counter. As the sum is below 1.0, the message is forwarded upstream, until it reaches node $u$ with $\xi_{ui}=0.9$. As the total sum is now $1.5>1.0$, the propagation over $p$ terminates, and a response is sent downstream by $u$. The response is shown in subfigure (b), accumulating the weights of edges it traverses. Nodes in its path, namely $v$ and $s$, sniff this information, as shown in subfigure (c), and collect measurements $t_{vi},t_{si}$ to be added to the averages estimating $\partial_{\xi_{ui}} L_\SR$ and $\partial_{\xi_{si}} L_\SR $, respectively. The source $s$ also collects measurement $t_{(i,s),p}=-t_{si}$, to be used in the average estimating  ${\partial_{\rho_{(i,s),p}} L} $.}\label{fig:controlexample}
\end{figure*}

\subsection{Subgradient Estimation}\label{sec:distributedsub} We now describe how to estimate the subgradients of $L_\SR$ through measurements collected during a timeslot. These estimates are computed in a distributed fashion at each node, using only information available from control messages traversing the node.
 Let $(\rho^{(k)},\Xi^{(k)})\in\conv(\feasibledomain_\SR)$ be the pair of global states at the $k$-th measurement period.
 At the conclusion of a timeslot, each $v\in V$ produces a random vector $z_{v}=z_{v}(\rho^{(k)},\Xi^{(k)})\in \reals_+^{|\catalog|}$ that is an unbiased estimator of a subgradient of $L_{\SR}$ w.r.t.~to $\xi_v$.  Similarly, for every $(i,s)\in \requests$, source node $s$ produces a random vector $q_{(i,s)} =q_{(i,s)}(\rho^{(k)},\Xi^{(k)}) \in \reals^{|\pathset_{(i,s)}|}$ that is an unbiased estimator of a subgradient of $L_{\SR}$ with respect to (w.r.t.)~$\rho_{(i,s)}$.  Formally, 
 \begin{align}\label{estimatezprop}
\expect[z_{v}(\rho^{(k)},\Xi^{(k)})] &\in \partial_{\xi_{v}}L_\SR(\rho^{(k)},\Xi^{(k)}), \\%, \quad  \expect\left[\|z_{v}(Y)\|_2^2\right]<M,
  \label{estimateqprop}
\expect[q_{(i,s)}(\rho^{(k)},\Xi^{(k)})]& \in \partial_{\rho_{(i,s)}}L_\SR(\rho^{(k)},\Xi^{(k)}),\end{align}
where $\partial_{\xi_{v}} L_\SR(\rho,\Xi)$, $\partial_{\rho_{(i,s)}}L_\SR$ are the sets of subgradients of $L_\SR$ w.r.t.~$\xi_{v}$ and $\rho_{(i,s)}$, respectively.
 To produce these estimates, nodes measure the upstream cost incurred at paths passing through it using control messages, exchanged among nodes as follows:\begin{packedenumerate}
\item  Every time a node $s$ generates a new request $(i,s)$, it also generates additional control messages, one per path $p\in \pathset_{(i,s)}$. The message corresponding to path $p$ is to be propagated over $p$, and contains a counter  initialized to  $1-\rho_{(i,s),p}+\xi_{si}$. \item When following path $p$, the message is forwarded until a node $u\in p$ s.t.~
$ 1-\rho_{(i,s),p} +\textstyle\sum_{\ell=1}^{k_{p}(u)} \xi_{p_{\ell}i}> 1 $ is found, or the end of the path is reached. To keep track of this, every $v\in p$ traversed adds its state variable  $\xi_{vi}$ to the message counter.  \item Upon reaching either such a  node $u$ or the end of the path, the control message is sent down in the reverse direction. Initializing its counter to zero, every time it traverses an edge in this reverse direction, it adds the weight of this edge into a weight counter. 
\item Every node on the reverse path ``sniffs'' the weight counter  of the control message, learning the sum of weights of all edges further upstream towards $u$; that is, recalling that $k_p(v)$ is the position of visited node $v\in p$, $v$ learns the quantity: \begin{align}\!\! t_{vi} = \textstyle\sum_{k'=k_v(p)}^{|p|-1} w_{p_{k'+1}p_{k'}} \id_{1-\rho_{(i,s),p}+\sum_{\ell=1}^{k'}  \xi_{p_{\ell}i}\leq 1}.\label{tvi}\end{align}
\item In addition, the source $s$ of the request, upon receiving the message sent over the reverse path, sniffs the quantity 
\begin{align}\!\!\!\!t_{(i,s),p}\! =\! -\!\textstyle\sum_{k'=1}^{|p|-1} w_{p_{k'+1}p_{k'}} \id_{1-\rho_{(i,s),p}+\!\sum_{\ell=1}^{k'} \! \xi_{p_{\ell}i}\leq 1}.\label{tisp}\end{align}
This is  the (negative of) the sum of weights accumulated by the control message returning to the source $s$.
  \end{packedenumerate}
An example illustrating the above five steps can be found in Figure~\ref{fig:controlexample}. Let $\mathcal{T}_{vi}$ be the set of quantities collected in this way at node $v$ regarding item $i\in \mathcal{C}$ during a measurement period of duration $T$. At the end of the timeslot, each node $v\in V$ produces $z_v$ as follows: \begin{align}z_{vi}= \textstyle\sum_{t\in \mathcal{T}_{vi} }t/T,\quad i\in\catalog.\label{zestimation}\end{align}
Similarly, let $\mathcal{T}_{(i,s),p}$ be the set of quantities collected in this way at source node $s$ regarding path $p\in \pathset_{(i,s)}$ during a measurement period of duration $T$. At the end of the measurement period, $s$ produces the  estimate $q_{(i,s)}$: \begin{align}q_{(i,s),p}= \textstyle\sum_{t\in \mathcal{T}_{(i,s),p} }t/T,\quad i\in\catalog.\label{qestimation}\end{align}
We show that the resulting $z_v$, $q_{(i,s)}$ satisfy \eqref{estimatezprop} and \eqref{estimateqprop}, respectively, in  Lemma~\ref{subgradientlemma}. 
In the above construction, control messages are sent over all paths in $\pathset_{(i,s)}$. It is important to note however that when sent over paths $p$ such that $\rho_{(i,s),p}\approx 0$ control messages \emph{do not travel far}: the termination condition (the sum exceeding 1) is satisfied early on. Messages sent over unlikely paths are thus pruned early, and ``deep'' propagation only happens in very likely paths.
Nevertheless, to reduce control traffic, in 
  Section~\ref{controlreduction} we   modify the algorithm to propagate only \emph{a single control message over a single path}.
\subsection{State Adaptation} 
 Having estimates $Z = [z_v]_{v\in V},q = [q_{(i,s)}]_{(i,s)\in\requests}$,  the global state  is adapted as follows: at the conclusion of the $k$-th period, the new state $(\rho^{(k+1)},\Xi^{(k+1)})$  is computed as:
 \begin{align}\mathcal{P}_{\!\conv(\feasibledomain_\SR)} \big(\rho^{(k)}\!\! +\! \gamma_k q(\rho^{(k)},\Xi^{(k)})\, ,\, \Xi^{(k)}\!\! +\!\gamma_k Z(\rho^{(k)},\Xi^{(k)})\big),\label{adapt}\end{align}
where $\gamma_k={1}/{\sqrt{k}}$ is a gain factor and $\mathcal{P}_{\conv(\feasibledomain_\SR)}$ is the orthogonal projection onto the convex set $\conv(\feasibledomain_\SR)$. Note that this additive adaptation and corresponding projection is separable across nodes and can be performed in a distributed fashion: each node $v\in V$ adapts its own relaxed caching strategy, each source $s$ adapts its routing strategy, and all nodes project these strategies to their respective local constraints implied by  \eqref{pselect},\eqref{xcapacity}, and the $[0,1]$ constraints. Note that these involve projections onto the rescaled simplex, for which well-known linear algorithms exist \citep{michelot1986finite}.
\subsection{State Smoothening.}
Upon performing the state adaptation \eqref{adapt}, each node $v\in V$ and each source $s$, for $(i,s)\in\requests$, compute the following
``sliding averages'' of current and past states:
 \begin{align}\label{slidexi}\bar{\xi}_v^{(k)} = \textstyle \sum_{\ell = \lfloor\frac{k}{2} \rfloor}^{k} \gamma_\ell \xi_v^{(\ell)} /\big[\sum_{\ell=\lfloor \frac{k}{2}\rfloor}^{k}\gamma_{\ell}\big]   .\\
 \label{sliderho}\bar{\rho}_s^{(k)} = \textstyle \sum_{\ell = \lfloor\frac{k}{2} \rfloor}^{k} \gamma_\ell \rho_v^{(\ell)} /\big[\sum_{\ell=\lfloor \frac{k}{2}\rfloor}^{k}\gamma_{\ell}\big]   .\end{align}
 This    is necessary  because of the non-differentiability of $L_\SR$ \citep{nemirovski2005efficient}. Note that    $(\bar{\rho}^{(k)},\bar{\Xi}^{(k)}) \in \conv( \feasibledomain_\SR) $, as a convex combination of elements of $\conv(\feasibledomain_\SR)$.

\subsection{Randomized Caching and Routing.} \label{sec:randomizedrounding}
The resulting $(\bar{\rho}^{(k)},\bar{\Xi}^{(k)})$  determine the randomized routing and caching strategies at each node during a timeslot.  First, given $\bar{\rho}^{(k)}$,  each time a request $(i,s)$ is generated, path $p\in \pathset_{(i,s)}$ is used to route the request with probability $\bar{\rho}_{(i,s),p}$, independently of past routing and caching decisions.  Second, given  $\bar{\xi}_v^{(k)}$, each node $v\in V$ reshuffles its contents, placing items in its cache independently of all other nodes: that is, node $v$ selects a random strategy  $x_v^{(k)}\in \{0,1\}^{|\catalog|}$ sampled independently of any other node in $V$.
 The  random strategy $x_v^{(k)}$ satisfies the following two properties:
\begin{packedenumerate}
\item It is a \emph{feasible} strategy, i.e., satisfies the capacity and integrality constraints \eqref{xcapacity} and \eqref{xintegrality}. 
\item It is \emph{consistent} with the marginals $\bar{\xi}_v^{(k)}$, i.e., for all $i\in\catalog,$
$\expect[x_{vi}^{(k)}\mid \bar{\xi}_{v}^{(k)}]=\bar{\xi}_{vi}^{(k)}.$ 
\end{packedenumerate}
We note that there can be many random caching strategies whose distributions  satisfy the above two properties. An efficient algorithm generating such a distribution is provided in \citep{wireless-cache} and, independently, in \citep{ioannidis2016adaptive}.  Given $\bar{\xi}_v^{(k)}$, a distribution over (deterministic) caching strategies can be computed
 in $O(c_v|\catalog|\log|\catalog|)$ time, and has $O(|C|)$ support; for the sake of completeness, we briefly outline this below. We follow the high-level description of \citep{wireless-cache} here; a detailed, formal description of the algorithm, a proof of its correctness, and a computational complexity analysis, can be found in \citep{ioannidis2016adaptive}.

\begin{figure}[t]
\hspace*{\stretch{1}}\includegraphics[width=0.6\columnwidth]{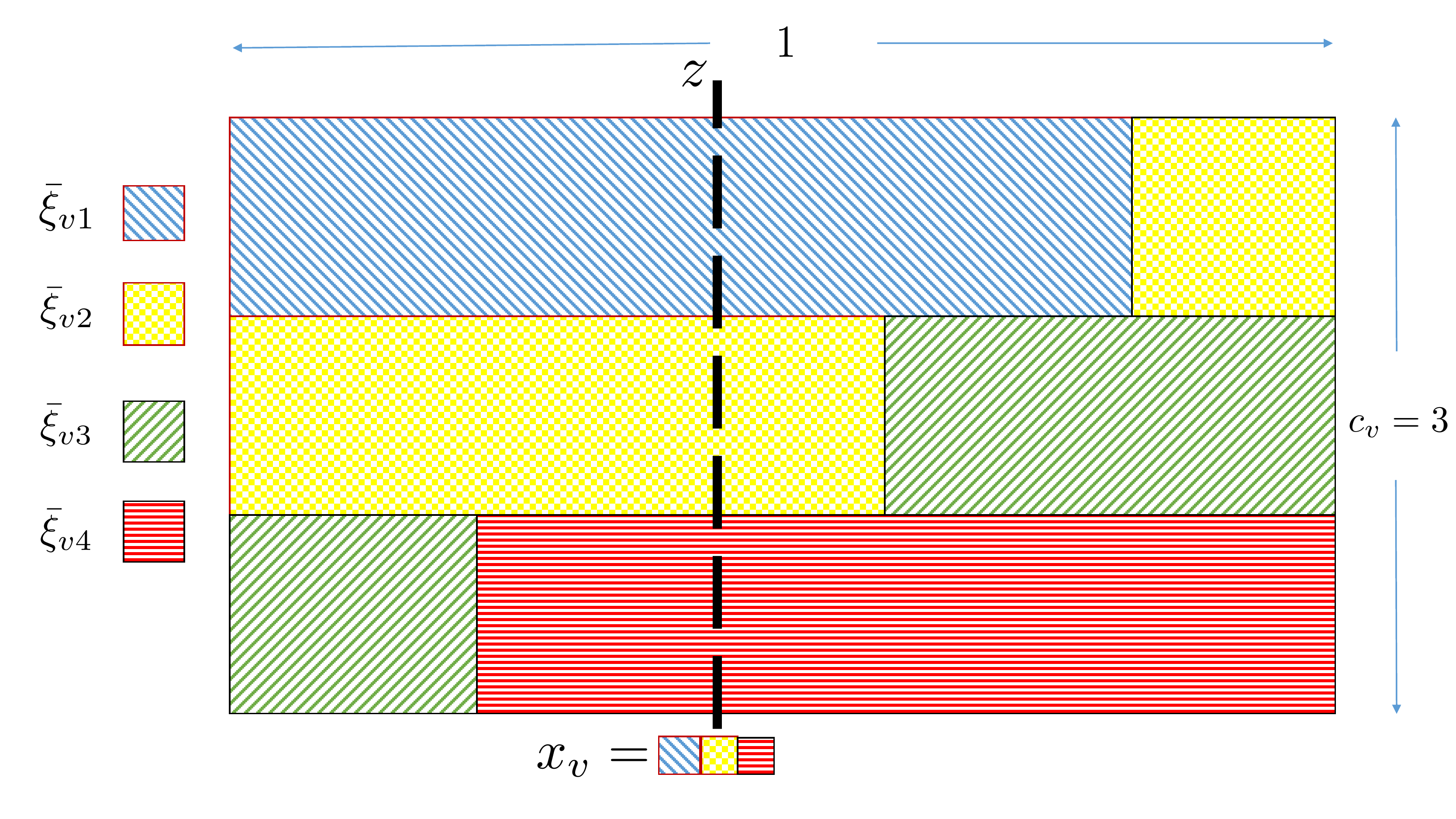}\hspace*{\stretch{1}}
\caption{Construction of a feasible randomized caching strategy $x_v$ that satisfies marginals $\prob[x_{vi}=1]=\bar{\xi}_{vi}$, where $\sum_{i\in \catalog}\bar{\xi}_{vi}=c_v$. In this example,  $c_v=3$, and $\catalog=\{1,2,3,4\}$. Given $\bar{\xi}_v$, 4 rectangles of height 1 each are constructed, such that the $i$-th rectangle has length $\bar{\xi}_{vi}\in[0,1]$, and the total length is $c_v$. After placing the 4 rectangles in a $3\times 1$ box, cutting the box at $z$ selected u.a.r. from $[0,1]$, and constructing a triplet of items from the rectangles it intersects, leads to an integral caching strategy with the desired marginals.}\label{fig:allocation}
 \end{figure}

The input to the algorithm are the marginal probabilities $\bar{\xi}_{vi}\in [0,1]$, $i\in \catalog$ s.t. $\sum_{i\in\catalog}\bar{\xi}_{vi}=c_v,$
where $c_v\in \naturals$ is the capacity of cache $v$. To construct a randomized caching strategy with the desired marginal distribution, consider  a rectangle box  of area $c_v\times 1$, as illustrated in Figure~\ref{fig:allocation}. For each $i\in \catalog$, place a rectangle of length $\bar{\xi}_{vi}$ and height 1 inside the box, starting from the top left corner. If a rectangle does not fit in a row, cut it, and place the remainder in the row immediately below, starting again from the left. As $\sum_{i\in \catalog}\bar{\xi}_{vi}=c_v$, this space-filling method  tessellates the $c_v\times 1$ box. The randomized placement then is constructed as follows: select a value in $z\in[0,1]$ uniformly at random, and ``cut'' the box at position $z$. The value will intersect exactly $c_v$ distinct rectangles:  as $\bar{\xi}_{vi}\leq 1$, no rectangle ``overlaps' with itself. The algorithm then produces as output the caching strategy $x_{v}\in \{0,1\}^{|\catalog|}$ where:
$$x_{vi}=\begin{cases}1, &\text{if the line intersects rectangle }i,\\
0, &\text{o.w.}\end{cases}$$
As the line intersects $c_v$ distinct rectangles, $\sum_{i\in\catalog}x_{vi}=c_v$, so the caching strategy is indeed feasible. On the other hand, by construction, the probability that $x_{vi}=1$ is exactly equal to the length of the $i$-th rectangle, so the marginal probability that $i$ is placed in the cache is indeed:
$$\prob[x_{vi}=1]=\bar{\xi}_{vi}.$$
and the randomized cache strategy $x_v$ has the desired marginals.

\subsection{Proof of Lemma~\ref{convergencelemma}}\label{proofofconvergencelemma}
We first show that \eqref{zestimation} and \eqref{qestimation} are unbiased estimators of the subgradient: \begin{lemma}\label{subgradientlemma}
The vectors  $z_{v}$, $v\in V$, and $q_{(i,s)}$, $(i,s)\in \requests$ constructed through coordinates \eqref{zestimation} and \eqref{qestimation},  satisfy properties \eqref{estimatezprop} and \eqref{estimateqprop}, respectively. Moreover
$\expect[\|z_v\|_2^2]<C_1,$ and $\expect[\|q_{(i,s)}\|_2^2]<C_2$,
where $$C_1=W^2\bar{P}^2|V|^2|\catalog|(\Lambda^2+\frac{\Lambda}{T}),~C_2=W^2|V|^2P(\Lambda^2+\frac{\Lambda}{T}),$$ and
$$W=\max_{(i,j)\in E}w_{ij},~\bar{P}=\max_{(i,s)\in \requests} |\pathset_{(i,s)}|,~\text{and}~\Lambda =\sum_{(i,s)\in \requests}\lambda_{(i,s)}.$$
\end{lemma}
\begin{proof}
 A vector $\zeta\in \reals^{|\catalog|}$  belongs to $\partial_{\xi_v}L_\SR(\rho,\Xi)$ if and only if
$\zeta_i \in [ \underline{\partial_{\xi_{vi}}L}_\SR(\rho,\Xi) ,\overline{\partial_{\xi_{vi}}L}_\SR(\rho,\Xi) ], $ where:
\begin{align*}
\begin{split}
\overline{\partial_{\xi_{vi}}L}_{\SR}(\rho,\Xi) &= \textstyle\sum_{(i,s)\in \requests}\lambda_{(i,s)}\sum_{p\in\pathset_{(i,s)}}\id_{v\in p} \cdot\\
& \textstyle\sum_{k'=k_p(v)}^{|p|-1}\!w_{p_{k'+1}p_{k'}} \id_{1-\rho_{(i,s)}+\sum_{\ell=1}^{k'} \xi_{p_{\ell}i}\leq 1},
\end{split}
\end{align*}
\begin{align*}
\begin{split}
\underline{\partial_{\xi_{vi}}L}_\SR(\rho,\Xi) & = \textstyle
\sum_{(i,s)\in \requests}\lambda_{(i,s)}\sum_{p\in\pathset_{(i,s)}}\id_{v\in p} \cdot\nonumber\\
& \textstyle \sum_{k'=k_p(v)}^{|p|-1} w_{p_{k'+1}p_{k'}} \id_{1-\rho_{(i,s)}+\sum_{\ell=1}^{k'} \xi_{p_{\ell}i}< 1}.
\end{split}
\end{align*}
If $L_\SR$ is differentiable at $(\rho,\Xi)$ w.r.t $\xi_{vi}$, the two limits coincide and are equal to $\tfrac{\partial L_\SR}{\partial _{vi}}.$
It immediately follows from the fact that requests are Poisson that $\expect[z_{vi}(\rho,\Xi)] = \overline{\partial_{\xi_{vi}}L}_\SR(\rho,\Xi)$, so indeed $\expect[z_{v}(Y)] \in \partial_{\xi_v}L_\SR(\rho,\Xi)$.
To prove the bound on the second moment, note that, for $T_{vi}$ the number of requests generated for $i$ that pass through $v$ during the slot,
$$\expect[z_{vi}^2]=\frac{1}{T^2}\expect[(\sum_{t\in \mathcal{T}_{vi}}t)^2] 
\leq \frac{W^2\bar{P}^2|V|^2}{T^2} \expect\Big[T_{vi}^2\Big],$$
as $t\leq W\bar{P}|V|$. On the other hand, $T_{vi}$ is Poisson distributed with expectation $$\sum_{(i,s)\in \requests}\id_{\exists p\in \pathset_{(i,s)}~\text{s.t.}~v\in p}\lambda_{(i,s)} T,$$ and the upper bound follows.
The  statement for $q_{(i,s)}$ follows similarly.  \end{proof}

 We now establish the convergence  of the smoothened marginals to a global maximizer of $L$.
\begin{comment} 
\begin{lemma} \label{convergencethm}Let $(\bar{\rho}^{(k)},\bar{\Xi}^{(k)})\in \feasibledomain_2$ be the smoothened state variables at the $k$-th period of Algorithm~\ref{alg:ascent}, and $$(\rho^*,\Xi^*)\in\argmax_{(\rho,\Xi)\in \conv(\feasibledomain_\SR)} L_\SR(\rho,\Xi).$$ Then, for $\gamma = 1/\sqrt{k}$,  $$\varepsilon_k\equiv\expect[L_\SR(\rho^*,\Xi^*)-L_\SR(\bar{\rho}^{(k)},\bar{\Xi}^{(k)})]=O(1/\sqrt{k}).$$
\end{lemma}
\end{comment}
 Under  \eqref{adapt}, \eqref{slidexi} and \eqref{sliderho}, from Theorem 14.1.1, page 215 of Nemirofski~\citep{nemirovski2005efficient}, $$\varepsilon_k\leq \frac{D^2 + M^2 \sum_{\ell=\lfloor k/2\rfloor}^{k}\gamma_\ell^2 }{2\sum_{\ell=\lfloor k/2\rfloor}^{k}\gamma_\ell},$$ 
where  $\gamma_k=\frac{1}{\sqrt{k}}$,
$$D\equiv\max_{x,y\in\conv(\feasibledomain_\SR)}\|x-y\|_2 = \sqrt{|V|\max_v 2 c_v  +2|\requests| },$$  and
$$M \equiv\sup_{(\rho,\Xi)} \sqrt{  \expect[\|Z(\rho,\Xi)\|_2^2] + \expect[\|q(\rho,\Xi)\|_2^2] }.$$    From Lem.~\ref{subgradientlemma}, 
$M \leq \sqrt{|V|C_1+|\requests|C_2},$ and the lemma follows. \hspace*{\stretch{1}}\qed

\subsection{Proof of Theorem~\ref{maincor}}\label{proofofmaincor}
By construction, conditioned on $(\bar{\rho}^{(k)},\bar{\Xi}^{(k)})$, the $|V|+|\requests|$ variables $x_v$,  $v\in V$, and $r_{(i,s)}$, $(i,s)$, are independent. Hence, conditioned on $(\bar{\rho}^{(k)},\bar{\Xi}^{(k)})$, all  monomial terms of $F_\SR$  involve independent random variables. Hence, 
$$\expect[F_\SR(r^{(k)},X^{(k)})\mid \bar{\rho}^{(k)},\bar{\Xi}^{(k)}] = F_\SR(\bar{\rho}^{(k)},\bar{\Xi}^{(k)}),$$ and, in turn,  
$$\lim_{k\to\infty} \expect[F_\SR(r^{(k)},X^{(k)})] =\lim_{k\to \infty }\expect[F_\SR(\bar{\rho}^{(k)},\bar{\Xi}^{(k)})].$$
Lemma~\ref{convergencelemma}
implies that, for $\nu^{(k)}$ the distribution of $(\bar{\rho}^{(k)}, \bar{\Xi}^{(k)})$, and $\Omega$ the set of $(\rho^*,\Xi^*) \in \conv(\feasibledomain_\SR)$ that are maximizers of $L_\SR$, 
 $$\lim_{k\to\infty} \nu^{(k)}(\conv(\feasibledomain_\SR)\setminus \Omega )=0.$$
By Lemma~\ref{goodapprox},   $$F_\SR(\rho^*,\Xi^*)\geq (1-1/e)\max_{(r,X)\in  \feasibledomain_\SR} F_\SR(r,X),$$ for any $(\rho^*,\Xi^*)\in \Omega$. The theorem follows from the above observations, and the fact that $F_\SR$ is bounded in $\conv(\feasibledomain_\SR)\setminus \Omega$.\qed

\fussy

\subsection{Reducing Control Traffic.}\label{controlreduction}
We now discuss how control messages generated by the protocol can be reduced by modifying the algorithm to propagate only a single control message over a single path with each request. The path is
 selected uniformly at random over paths in the support of $\rho_{(i,s)}$. 
That is, when a request $(i,s)\in\requests$ arrives at $s$, a single control message  is propagated over $p$ selected uniformly at random from 
$\supp(\rho_{(i,s)})=\{p \in \pathset_{(i,s)}: \rho_{(i,s),p}>0\}.$
This reduces the number of control messages generated by $s$ by  at least a $c=|\supp(\rho_{(i,s)})|$ factor. To ensure that \eqref{estimatezprop} and \eqref{estimateqprop} hold, it suffices to rescale measured upstream costs by $c$. To do this, the (single) control message contains an additional field storing $c$. When extracting weight counters from downwards packets, nodes on the path compute $t'_{vi}  =  c \cdot t_{vi}$, and $t'_{(i,s),p} = c \cdot t_{(i,s)p}$,
where   $t_{vi}$, $t_{(i,s),p}$ are as in \eqref{tvi} and \eqref{tisp}, respectively. This randomization reduces   control traffic, but  increases the \emph{variance} of  subgradient estimates, also by a factor of $c$. This, in turn, slows down the algorithm convergence; this tradeoff can be quantified  through, e.g., the constants in Lemma~\ref{convergencelemma}.

\section{Hop-By-Hop Routing}\label{app:hopbyhop}

The proofs for the hop-by-hop setting are similar, \emph{mutatis-mutandis}, as the proofs of the source routing setting.
As such, in our exposition below, we focus on the main technical differences between the algorithms for the two settings.

\subsection{Proof of Thm~\ref{thm:offlinehop} (Sketch)}
A constant approximation algorithm for solving \CGHH can be constructed following the same steps as for source routing. Consider the function
 $L_{\HH}:\conv(\feasibledomain_{\HH})\to \reals_+$, defined as:
\begin{align*}
\begin{split}
L_{\HH} (\rho,\Xi) = \textstyle\sum_{(i,s) \in \requests }\lambda_{(i,s)}\sum_{(u,v)\in G^{(i,s)}}\!\sum_{p\in \pathset_{(i,s)}^u} \!\!w_{vu} \cdot\\
\textstyle\min\{1, 1\!-\!\rho_{uv}^{(i)}\!+\!\xi_{ui}\!+\! \sum_{k'=1}^{|p|-1}(1\!- \!\rho_{p_{k'}p_{k'+1}}^{(i)}\!+\!\xi_{p_{k'}i})\big\}
\end{split}
\end{align*}
As in Lemma~\ref{goodapprox}, we can show that this concave function approximates $F_{\HH}$, in that for all   $(\rho,\Xi)\in \conv(\feasibledomain_\HH):$
\begin{align}
(1-{1}/{e})  L_{\HH} (\rho,\Xi) \leq F_{\HH}(\rho,\Xi)\leq L_{\HH} (\rho,\Xi).\end{align}

To construct a constant approximation solution, first, a fractional solution
\begin{align}(\rho^*,\Xi^*)=\argmax_{(\rho,\Xi)\in \conv(\feasibledomain_\HH) } L_{\HH} (\rho,\Xi),\label{convsolve2} \end{align}
 can be obtained. This again involves a convex optimization, which again can be reduced (as in Section~\ref{app:linear}) to a linear program. Subsequently, the solution can be rounded to obtain  an integral solution $(r,X)\in \feasibledomain_\SR$ such that $F_\SR(r,X)\geq F_\SR(\rho^*,\Xi^*)$.

Rounding follows the same steps as for source routing,  with an exception for rounding the routing strategy $\rho^*$. To round $\rho^*$, one first rounds each node's strategy individually, i.e., for every $v\in V$ and every $i\in \catalog$, we would pick the neighbor that maximizes the objective. This again follows from the fact that, given a caching strategy $\Xi$, and given the routing strategies of all other nodes, the objective is an affine function of $\{r_{uv}^{(i)}\}_{v:(uv)\in E^{(i)}}$, for all $u\in V$, with positive coefficients. Hence, keeping everything else fixed, if each node chooses a cost minimizing decision, this will round its strategy, and all nodes in $V$ can do this sequentially.
Note that, contrary to source routing, the order with which this rounding happens affects the final integral solution. Moreover, the DAG property  ensures that all requests eventually reach a designated server, irrespectively of the routing strategies resulting from the rounding decisions.

\subsection{Proof of Theorem~\ref{maincor2} (Sketch)}

A distributed algorithm can be constructed by performing projected gradient ascent over $L_\HH$. Beyond the same caching state variables $\xi_v$ stored at each node $v\in V$, each node $v\in V$ maintains routing state variables $\rho_{u}^{(i)} = [ \rho_{uv}^{(i)}]_{v:(u,v)\in E^{(i)}}\in [0,1]^{|E^{(i)}|},$ for each $i\in \catalog$,
containing  the marginal probabilities $\rho_{uv}^{(i)}$ that $u$ routes request message for item $i$ towards $v\in E^{(i)}$. Time is slotted, and nodes perform  subgradient estimation, state adaptation, state smoothening, and randomized sampling of caching and routing strategies. As the last three steps are nearly identical to source routing, we focus below on how to estimate subgradients, which is the key difference between the two algorithms.

Whenever a request $(i,s)\in \requests$ is generated, a control message is propagated in all neighbors of $s$ in $E^{(i)}$. These messages contain counters initialized to
$1-\rho_{sv}^{(i)} + \xi_{si}. $
Each node $v\in V$ receiving such a message generates one copy for each of its neighbors in $E^{(i)}$. For each neighbor $u$, $v$ adds $1-\rho_{vu}^{(i)} +\xi_{vi}$ to the counter, and forwards the message to $u$ if the counter is below 1.0.
Formally,  a control message originating at $s$ and reaching  a node $v$ after having followed path $p\in G^{(i,s)}$ is forwarded to $u$ if the following condition is satisfied:
\begin{align}  1-\rho_{vu}^{(i)}+\xi_{ui} +\textstyle\sum_{\ell=1}^{k_{p}(v)-1} \big(1-\rho_{p_\ell p_{\ell+1}}^{(i)} +\xi_{p_{\ell}i}\big)\leq 1\label{stopcrit}\end{align}
If this condition is met,  $v$ forwards a copy of the control message to $u$; the  above process is repeated at each of each neighbors. If the condition fails \emph{for all} neighbors, a response message is generated by $v$ and propagated over the reverse path, accumulating the weights of edges it passes through. Moreover, descending control messages are merged as follows. Each node $v$ waits for all responses from neighbors to which it has sent control messages; upon the  last arrival, it adds their counters, and sends the ``merged'' message containing the accumulated counter reversely over path $p$.

As before, messages on the return path are again ``sniffed'' by nodes they pass through, extracting the upstream costs. Their averages are used as estimators of the subgradients w.r.t. both the local routing and caching states, in a manner similar to how this was performed in source routing. 
\begin{comment}
 In particular, the upstream control message received by $v$ from edge $(v,u)\in E^{(i)}$ contains the value:
\begin{align*}\begin{split}t^{(i)}_{uv}\equiv \sum_{(v',u')\in G^{(i,s)}} \sum_{\substack{p\in \pathset^{v}_{(i,s)}:\\(v,u)\in p}} w_{u'v'}\cdot\\ \qquad\id_{1- \rho_{v'u'}^{(i)}+\xi_{v'i}+\sum_{\ell=1}^{k_{p}(v')-1} \big(1-\rho_{p_\ell p_{\ell+1}^{i}} +\xi_{p_{\ell}i}\big)},
\end{split}
\end{align*}
which contributes to an average used for estimating of $\partial_{\rho_{vu}^{(i)}}L_\HH(\rho,\Xi)$. Similarly, the sum of all values coming from all neighbors, i.e.,
$t_{vi} \equiv \sum_{(v,u)\in E^{(i)}} t^{(i)}_{uv}$
 contributes to the average used for estimating $\partial_{\xi_{vi}}L_\HH(\rho,\Xi)$. 
\end{comment}
As each edge is traversed at most twice, the maximum number of control messages is $O(|E^{(i)}|).$ As in the case of source routing, however, messages on low-probability paths are pruned early on due to condition \eqref{stopcrit}. Moreover, as in Section~\ref{controlreduction}, randomization can again reduce the number of messages:  only a single message need be propagated to a neighbor selected uniformly at random; in this case, the message needs to also contain a field keeping track of the product of the size of neighborhoods of nodes it has passed through, and updated by each node by multiplying the entry by the size of its own neighborhood. As in source routing, this is used as an additional scaling factor for quantities $t^{(i)}_{vu}$, $t_{vi}$.

\section{Detailed Description of Graphs in Table~{\protect\lowercase{\ref{networks}}}}

\sloppy
Graph \texttt{cycle} is a simple cyclic graph; \texttt{grid-2d} is a two-dimensional square grid;  \texttt{hypercube} is a 7-dimensional hypercube. Graph \texttt{expander} is a Margulies-Gabber-Galil expander \citep{gabber1981explicit}.
The next 5 graphs are random, i.e., were sampled from a probability distribution. Graph \texttt{erdos-renyi} is an Erd\H{o}s-R\'enyi graph with parameter $p=0.1$, and \texttt{regular} is a 3-regular graph sampled uniformly at random (u.a.r.). The \texttt{watts-strogatz} graph is  generated according to the Watts-Strogatz model of a small-world network \citep{watts1998collective}, i.e., a cycle and $4$ randomly selected edges, while \texttt{small-world} is the graph by Kleinberg \citep{kleinberg2000small}, i.e., a grid with additional long range edges. The preferential attachment model of Barab\'asi and Albert \citep{barabasi1999emergence}, which yields powerlaw degrees, is used for \texttt{barabasi-albert}.
The last 3 graphs are the GEANT, Abilene, and  Deutche Telekom backbone networks \citep{rossi2011caching}.

\fussy

\end{document}